\documentclass[pdflatex,sn-mathphys-num]{sn-jnl}

\usepackage{graphicx}%
\usepackage{multirow}%
\usepackage{amsmath,amssymb,amsfonts}%
\usepackage{amsthm}%
\usepackage{mathrsfs}%
\usepackage[title]{appendix}%
\usepackage{xcolor}%
\usepackage{textcomp}%
\usepackage{manyfoot}%
\usepackage{booktabs}%
\usepackage{algorithm}%
\usepackage{algorithmicx}%
\usepackage{algpseudocode}%
\usepackage{listings}%
\usepackage{bm}

\usepackage{thmtools,thm-restate}
\usepackage{tikz}
\usepackage{braket}
\usetikzlibrary{positioning}

\newcommand{\K}{K}                    
\newcommand{\alphavec}{\boldsymbol{\alpha}}
\newcommand{\etak}{\eta_{\K}}         

\declaretheorem[name=Lemma,numberwithin=section]{lemma}

\theoremstyle{thmstyleone}%
\newtheorem{theorem}{Theorem}
\newtheorem{proposition}[theorem]{Proposition}%

\theoremstyle{thmstyletwo}%
\newtheorem{remark}{Remark}%

\theoremstyle{thmstylethree}%
\newtheorem{definition}{Definition}%
\declaretheorem[name=Corollary,numberwithin=section]{corollary}

\begin{document}

\title[Article Title]{Kernelized Decoded Quantum Interferometry}


\author*[1,2]{\fnm{Fumin} \sur{Wang}}\email{Contact author: fwang1991@xjtu.edu.cn}

\affil*[1]{\orgdiv{MED-X Institute}, \orgname{the First Affiliated Hospital of Xi’an Jiaotong University}, \orgaddress{ \city{Xi'an}, \postcode{710061}, \country{China}}}

\affil[2]{\orgdiv{Shaanxi Key Laboratory of Quantum Information and Quantum Optoelectronic Devices, College of Physics}, \orgname{Xi'an Jiaotong University}, \city{Xi'an}, \postcode{710049}, \country{China}}


\abstract{Decoded Quantum Interferometry (DQI) promises superpolynomial speedups for structured optimization; however, its practical realization is often hindered by significant sensitivity to hardware noise and spectral dispersion. To bridge this gap, we introduce Kernelized Decoded Quantum Interferometry (k-DQI), a unified framework that integrates spectral engineering directly into the quantum circuit architecture. By inserting a unitary kernel prior to the interference step, k-DQI actively reshapes the problem's energy landscape, concentrating the solution mass into a ``decoder-friendly'' low-frequency head. We formalize this advantage through a novel robustness metric, the noise-weighted head mass $\Sigma_K$, and prove a Monotonic Improvement Theorem, which establishes that maximizing $\Sigma_K$ guarantees higher decoding success rates under local depolarizing noise. We substantiate these theoretical gains in Optimal Polynomial Interpolation (OPI) and LDPC-like problems, demonstrating that kernel tuning functions as a ``spectral lens'' to recover signal otherwise lost to isotropic noise. Crucially, we provide explicit, efficient circuit realizations using Chirp and Linear Canonical Transform (LCT) kernels that achieve significant boosts in effective signal-to-noise ratio with negligible depth overhead ($\tilde{O}(n)$ to $\tilde{O}(n^2)$). Collectively, these results reframe DQI from a static algorithm into a tunable, noise-aware protocol suited for near-term error-corrected environments.}

\maketitle

\section{Introduction}\label{sec1}

Quantum algorithms for discrete optimization face a fundamental dichotomy: while worst-case instances are likely intractable, practical instances often possess \textit{hidden structures}—such as algebraic symmetries or sparse dependency graphs—that theoretically permit superpolynomial speedups. Decoded Quantum Interferometry (DQI) \cite{JordanNature2025} has emerged as a powerful paradigm to exploit these structures. By encoding objective values into quantum amplitudes and subjecting them to a Fourier or Walsh–Hadamard transform \cite{NielsenChuang2010,Coppersmith2002AQFT}, DQI attempts to concentrate the solution's “energy’’ into low-frequency modes that can be efficiently recovered by classical decoders, such as belief propagation or Reed–Solomon decoding \cite{Kschischang2001Factor,RichardsonUrbanke2001,Gallager1962LDPC,MacKay1999Good,Sudan1997RS,GuruswamiSudan1999,Koetter2003ASD}. However, standard DQI relies on a fragile coincidence: the problem's native spectral geometry must accidentally align with the decoder's constrained capability. In the presence of hardware noise or slight structural deviations, this alignment breaks, and the “signal’’ diffuses irretrievably into the high-frequency tail \cite{Anschuetz2025RequiresStructure,Preskill2018NISQ,McClean2018Barren,Wang2021NIBP}. This limitation poses a critical question: Can we systematically reshape the spectrum within the quantum circuit to guarantee robustness, rather than passively hoping for favorable instance presentation?

In this work, we answer affirmatively by introducing Kernelized DQI (k-DQI), a unified framework that inserts an analytically tractable, efficiently realizable unitary \textit{kernel} $K$ before the interference step. Conceptually, k-DQI acts as a spectral preconditioner. Just as an optical lens focuses light to overcome transmission loss, the kernel $K$ actively remaps the objective function, concentrating its Walsh/Fourier energy onto supports that known decoders handle reliably \cite{Kschischang2001Factor,RichardsonUrbanke2001,Gallager1962LDPC,MacKay1999Good}. This architectural shift effectively converts DQI’s often implicit “needs structure’’ requirement into an explicit design degree of freedom: we no longer ask binary questions like “Does this instance have structure?’’; instead, we ask “Which kernel $K$ best exposes the structure?’’

We will refer to this use of $K$ as spectral preconditioning. It is helpful to view our setting as a reusable design template: (i) choose a shaping map $P(f)$ that embeds the objective into a spectrum compatible with a fixed interferometer $F$; (ii) insert a parameterized kernel family $K(\theta)$ in front of $F$; and (iii) tune $\theta$ so as to maximize a simple spectral statistic such as the noise–weighted head mass $\Sigma_K(\ell,\eta;d)$. Therefore, the fragility of standard DQI can be phrased as a precise design question. Given a shaped spectrum $\widehat\psi$ and a noisy channel acting after the interference layer, the native spectral geometry must align with the capabilities of a restricted decoder family. In practice this alignment is often destroyed by hardware imperfections or slight instance perturbations \cite{Preskill2018NISQ}. Our goal in this work is to replace this incidental alignment by an explicit, tunable degree of freedom: we seek efficiently implementable unitary kernels $K$ such that, under a realistic noise model and for a given decoder family, the resulting noise-weighted head mass $\widetilde H_k$ is provably increased and the corresponding approximation-ratio guarantees are improved. In other words, rather than asking whether an instance “has structure,’’ we ask which kernels expose sufficient spectral structure to place the system above a known decoding threshold \cite{RichardsonUrbanke2001}.

We substantiate this approach through three primary contributions that bridge quantum spectral analysis with classical coding theory. First, we establish the Monotonic Improvement Theorem. We define the noise-weighted head-spectrum mass $\Sigma_{K}(\ell,\eta)$, a rigorous metric quantifying the effective signal-to-noise ratio entering the decoder, and prove that under local depolarizing noise \cite{Preskill2018NISQ} any increase in $\Sigma_{K}$ strictly improves the k-DQI approximation ratio. This result provides a “calculus of robustness,’’ ensuring that spectral engineering translates directly into algorithmic success.

Second, we demonstrate structure-adaptive guarantees across distinct problem classes. For Optimal Polynomial Interpolation (OPI), we utilize chirp and Linear Canonical Transform (LCT) kernels to diagonalize algebraic symmetries \cite{Namias1980FRFT,Healy2016LCTBook,Rabiner1969CZT}, allowing k-DQI to match or exceed the performance of classical soft/list decoders \cite{Sudan1997RS,GuruswamiSudan1999,Koetter2003ASD}. For sparse Max-XORSAT, we design block-local kernels that preserve graph locality while strictly increasing the belief-propagation threshold (as verified by density evolution) \cite{Kschischang2001Factor,RichardsonUrbanke2001}, effectively pushing the system back across the phase-transition boundary from “unsolvable’’ to “solvable.’’

Finally, we show that this robustness comes at negligible computational cost. We provide explicit circuit constructions demonstrating that the kernel layer adds only linear or near-linear diagonal/quadratic-phase gates to the pipeline \cite{NielsenChuang2010,Coppersmith2002AQFT}. Thus, k-DQI achieves its gains purely through better spectral alignment rather than by increasing circuit depth. Collectively, these results reframe DQI from a static procedure into a tunable, noise-resilient protocol, laying the groundwork for high-dimensional generalizations governed by verifiable information-theoretic limits rather than ad-hoc heuristics.

\section{Results}\label{sec:k-dqi}

\subsection{Kernelized DQI: model and guarantees}

The conceptual workflow of the proposed Kernelized Decoded Quantum Interferometry (k-DQI) framework is depicted in Fig.~\ref{fig1}. 
Unlike standard approaches that rely solely on the intrinsic structure of the problem, our protocol introduces an explicit spectral engineering layer. Usefully viewed as a three-stage pipeline, it proceeds by first shaping the objective function $f$ into a polynomial representation $P(f)$, followed by a quantum evolution where a unitary kernel $K$ actively reshapes the probability distribution before noise exposure. Finally, the noisy measurement outcomes are processed by a classical decoder to recover the candidate solution $x^*$. This modular design allows us to theoretically quantify and optimize the robustness of the algorithm via the kernel spectrum $\Sigma_K$.

Let $f : \{0,1\}^n \to \mathbb{R}$ denote an objective function (e.g., Max-LINSAT/XORSAT or OPI after encoding on qubits). 
Fix a degree-$\ell$ polynomial $P$ with bounded coefficients.
Following DQI \citep{JordanNature2025}, the shaped state prepared by the quantum circuit is
\begin{equation}
  \label{eq:shaped-state}
  \ket{\psi_f}
  = \sum_{x\in\{0,1\}^n} P\!\bigl(f(x)\bigr)\ket{x}.
\end{equation}
In k-DQI we choose a unitary kernel $K\in U(2^n)$, apply $K$ to the shaped state, then apply the interferometer $F$ and finally decode the measurement outcomes.
We write $F = H^{\otimes n}$ for the Walsh--Hadamard transform on $n$ qubits (or the appropriate QFT over $\mathbb{F}_p^n$ for $p$-ary encodings).
Measurement outcomes in the computational basis are interpreted by a (coherent or classical) decoder $\mathrm{Dec}$, producing a candidate solution $x^\star$; the usual DQI reduction maps the objective value achieved by $x^\star$ to a decoding success event.

It is convenient to collect the shaping and kernel into a single spectral object.
Let $g(x) = P(f(x))$ and $\|g\|_2^2 = \sum_x |g(x)|^2$.
We define the normalized shaped state
\begin{equation}
  \ket{\phi_f} = \frac{1}{\|g\|_2} \sum_x g(x)\ket{x},
\end{equation}
and denote by
\begin{equation}
  \label{eq:alpha-def}
  \alpha := \frac{(F K)g}{\|g\|_2} \in \mathbb{C}^{2^n}
\end{equation}
the spectrum after applying the kernel and the interferometer.
All of the decoding behavior of k-DQI will be expressed in terms of this vector $\alpha$.

Throughout we choose $K$ from an \emph{admissible} kernel family: a collection $\{K(\theta)\}_{\theta\in\Theta}\subset U(2^n)$ that
(i) admits poly$(n,\log(1/\varepsilon))$-depth circuit realizations,
(ii) preserves factor-graph locality up to bounded-width mixing for sparse instances, and
(iii) contains the identity $K=I$.
For technical reasons, we also require the head-to-tail block of $F K$ to be strictly contractive in operator norm for the head sizes $d$ of interest.
Examples of admissible families include discrete linear canonical transforms and chirp-phase maps that approximately diagonalize Vandermonde-like structure for OPI instances and respect bounded-width locality for sparse LDPC-like instances; see Sections~\ref{subsec:opi} and~\ref{subsec:ldpc}.

The kernel $K$ plays the role of a spectral preconditioner.
After shaping and interference, the amplitudes $\alpha_s$ in~\eqref{eq:alpha-def} determine how measurement outcomes populate the various spectral components.
For any head size $d$, let $S_d(\alpha)$ denote the indices of the $d$ largest $|\alpha_s|$.
All decoders we consider are effectively sensitive only to this ``head'' set: they succeed whenever a sufficiently large fraction of the total spectral weight lies on decoder-friendly supports, and they fail when the head is too small or too noisy.

To capture this behavior we introduce a \emph{noise-weighted head-spectrum mass} $\Sigma_K(\ell,\eta;d)$.
For a head size $d$, let $S_d(\alphavec)$ be the indices of the $d$ largest $|\alpha_s|$.
Given per-mode attenuation $\etak(s)\in(0,1]$ (from the noise model) we define the
\emph{normalized noise-weighted head mass}
\begin{equation}
\Sigma_K(\ell,\eta;d):=\sum_{s\in S_d(\alphavec)} \etak(s)\,|\alpha_s|^2\in[0,1].
\label{eq:SigmaK_def_norm}
\end{equation}
It coincides with the usual head mass when $\etak\equiv 1$, is nondecreasing in $d$, and is
invariant under global phase.
Roughly speaking, $\Sigma_K(\ell,\eta;d)$ measures how much of the shaped-and-kernelized spectrum survives in the head $S_d(\alpha)$ after passing through a noisy channel with effective per-mode attenuations $\eta_K(s)\in(0,1]$.
A larger value of $\Sigma_K$ means that more probability mass reaches the decoder in a form it can reliably handle.
The main message of this section is that, under a mild monotone-threshold property on the decoder, \emph{any increase in $\Sigma_K$ leads to a corresponding improvement in the k-DQI approximation guarantees}.
Theorem~\ref{thm:monotonicity} below makes this statement precise.

\subsection{The Monotonic Improvement Theorem}
\label{sec:monotone-theorem}
Our central theoretical contribution is establishing that the spectral head mass $\Sigma_K$ is not merely a heuristic metric, but a rigorous proxy for algorithmic success under noise.

\begin{theorem}[Monotonic Improvement via Head Mass]
\label{thm:monotonicity}
Fix the shaping degree $\ell$ and noise parameters (rate $\eta$, mixing width $w$). The k-DQI approximation ratio lower bound $A(\ell, \eta, K; d)$ is a strictly non-decreasing function of the noise-weighted head mass $\Sigma_K(\ell, \eta; d)$. 

Formally, if a kernel $K'$ yields a strictly larger head mass than $K$ (i.e., $\Sigma_{K'} > \Sigma_K$) and operates above the decoding threshold, then:
\begin{equation}
    A(\ell, \eta, K'; d) > A(\ell, \eta, K; d).
\end{equation}
\end{theorem}

\textit{Physical Intuition \& Implication:} 
Theorem \ref{thm:monotonicity} transforms the algorithm design problem from an abstract search for "quantum advantage" into a concrete optimization task: \textbf{maximize the spectral energy $\Sigma_K$}. 
Intuitively, the kernel $K$ acts as a "spectral lens." Just as an optical lens focuses light to overcome transmission loss, a well-designed kernel concentrates the problem's solution mass into the few low-frequency modes (the "head") that classical decoders can reliably recover, even after attenuation by channel noise.
(See Appendix~A for the full operator-norm derivation and noise contraction bounds.)

\begin{figure}[t]
  \centering
  \includegraphics[width=1\linewidth]{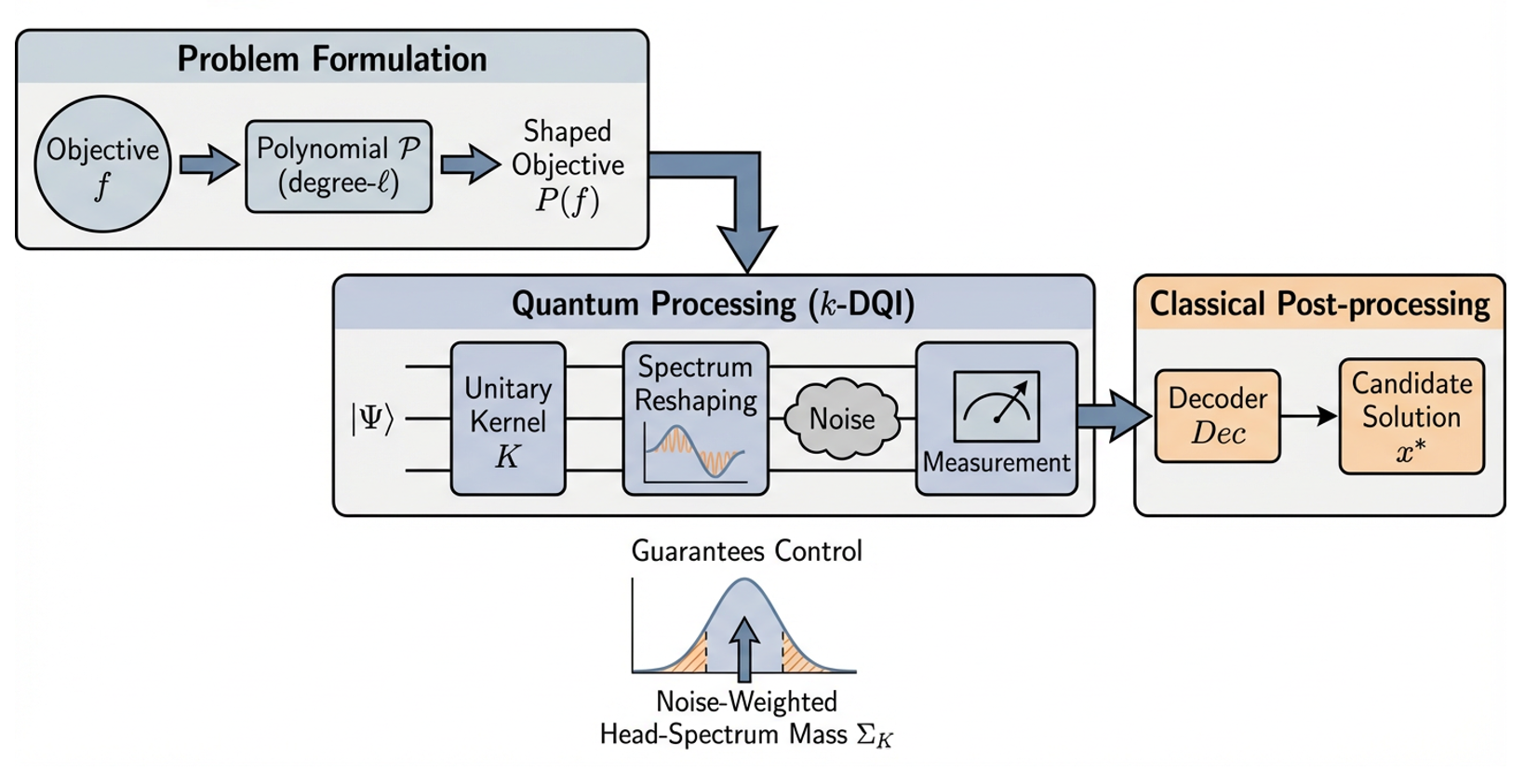}
  \caption{Pipeline of kernelized decoded quantum interferometry (k-DQI).
  A degree-$\ell$ polynomial $P$ first shapes the objective $f$, a unitary kernel $K$
  reshapes the resulting spectrum before noise, and a decoder $\mathrm{Dec}$ maps
  measurement outcomes to a candidate solution $x^\star$.
  The kernel is designed to increase a noise-weighted head-spectrum mass $\Sigma_K$ that
  controls the resulting guarantees.}
  \label{fig1}
\end{figure}

To visualize the mechanism underpinning Theorem~II.1, Fig.~\ref{fig:bec-de} plots the BEC density–evolution map
\(y=\varepsilon\,\lambda(1-\rho(1-x))\) for a $(3,6)$ ensemble together with the diagonal \(y=x\).
Replacing \(\varepsilon\) by the effective \(\varepsilon'=\varepsilon(1-\kappa\,\Delta\Sigma_{\mathrm{loc}})\) shifts the fixed point below the diagonal,
thereby certifying convergence at a strictly larger operating point.
This surrogate picture is exactly the monotone–improvement statement of Theorem~II.1:
any increase in the noise–weighted head mass raises the decoder success proxy and hence
improves the threshold.%

\begin{figure}[t]
  \centering
  \includegraphics[width=0.95\linewidth]{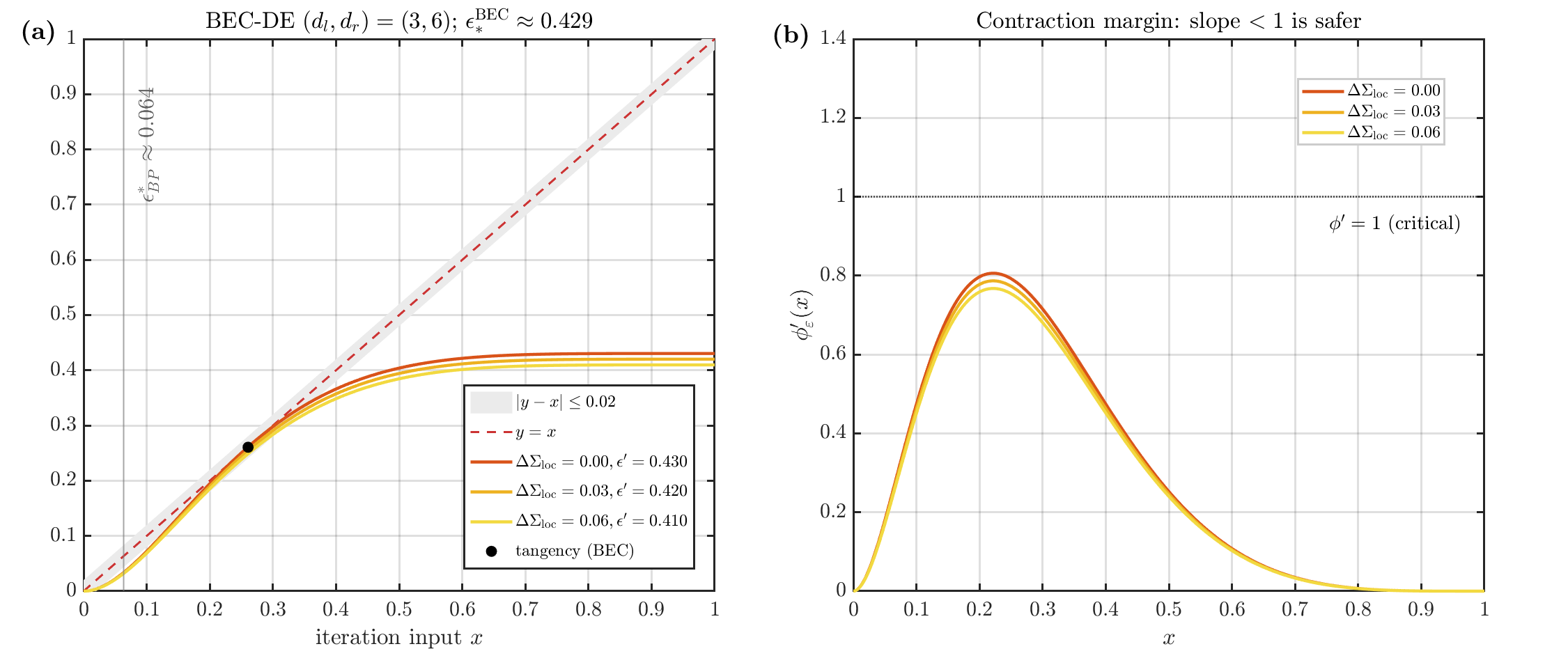}
  \caption{\textbf{Mechanism of threshold improvement via block-local alignment (BLA).} 
    The figure illustrates the behavior of the density evolution surrogate for a regular $(d_l,d_r)=(3,6)$ LDPC ensemble over the Binary Erasure Channel (BEC).
    \textbf{a}, The recursive DE map $y = \phi_\epsilon(x)$. The black point marks the critical ``tangency condition'' at the asymptotic BEC threshold ($ \epsilon^*_{\mathrm{BEC}} \approx 0.429$), where the recursion gets stuck. 
    The colored curves demonstrate the dose-response effect of the kernel structure: increasing the local alignment $\Delta\Sigma_{\mathrm{loc}}$ from $0$ (dark orange) to $0.06$ (yellow) lowers the effective noise parameter $\epsilon' = \epsilon(1 - \kappa \Delta\Sigma_{\mathrm{loc}})$. This shift pushes the curve strictly below the diagonal $y=x$ (red dashed line), opening a ``tunnel'' for successful decoding convergence. The gray band denotes the visual tolerance margin $|y-x| \le 0.02$.
    \textbf{b}, Contraction stability analysis via the derivative $\phi'_\epsilon(x)$. Reliable belief propagation requires the slope to remain below unity (dotted line at 1.0). The plot confirms that the BLA-induced shift dampens the derivative peak, restoring the strict contraction property ($\phi' < 1$) required to bypass the bottleneck.}
  \label{fig:bec-de}
\end{figure}

\subsection{Noise model and decoder property}
\label{subsec:noise-decoder}

We adopt the standard DQI noise model used in recent analyses~\citep{Bu2025Noise}, consisting of local depolarizing noise, an optional banded unitary mixing layer, and nonuniform loss.  After the kernel $K$ and interferometer $\mathcal{F}$, the overall channel acts as
\begin{equation}
  \mathcal{E} = \mathcal{L} \circ \mathcal{M} \circ \mathcal{Z},
\end{equation}
where $\mathcal{Z}$ is local depolarizing of rate $\eta \in [0,1)$ on each qubit, $\mathcal{M}$ is a banded unitary of width $w$ that mixes neighboring spectral components, and $\mathcal{L}$ is diagonal loss with transmittances $\tau(s)\in(0,1]$.  Under a small-angle/weak-mixing approximation, these three components combine into effective per-mode attenuation factors $\eta_{K}(s)\in(0,1]$ and a leakage remainder $\Delta(w)$; a precise inequality and its proof are given in Lemma~\ref{lem:noise} in Appendix~\ref{app:noise-model}.

On the decoding side we only require a mild monotone-threshold property that abstracts algebraic soft/list decoders on RS-like families and BP/density-evolution guarantees on LDPC-like families~\citep{GuruswamiSudan1999,KoetterVardy2003,RichardsonUrbanke2001,RichardsonShokrollahiUrbanke2001,Kschischang2001Factor}.   We assume that when a decoder $\mathrm{Dec}$ is fed with measurement outcomes supported on a head set $\mathcal{S}_d(\alpha)$ under local noise weights $\eta_{K}$, its success probability in achieving a target approximation ratio $\rho\in(0,1]$ is a nondecreasing function of the noise-weighted head-spectrum mass $\Sigma_{K}(\ell,\eta;d)$.  Concretely, there exists a nondecreasing response curve $\Phi$ such that
\begin{equation}
  \Pr\!\bigl[\mathrm{Dec}\ \text{returns a solution of value}\ \ge \rho\bigr]
  \;\ge\;
  \Phi\!\bigl(\Sigma_{K}(\ell,\eta;d),\, d,\, \text{ensemble parameters}\bigr),
\end{equation}
where the mapping from decoder success to objective value follows the standard DQI reduction~\citep{JordanNature2025}.  For RS-like (OPI) instances, $\Phi$ can be made explicit using BM/KV/GS decoders~\citep{GuruswamiSudan1999,KoetterVardy2003}, and for LDPC-like ensembles it is given by density-evolution or EXIT-style analyses~\citep{RichardsonUrbanke2001,RichardsonShokrollahiUrbanke2001,Kschischang2001Factor}; the corresponding formulas and coding-theoretic justification are collected in Appendix~A.

Admissible kernel families include discrete linear canonical transforms and chirp‑phase maps that (i) can be realized by poly‑depth circuits, (ii) approximately diagonalize Vandermonde‑like structure for OPI, and (iii) respect bounded‑width locality for sparse LINSAT/XORSAT so that density‑evolution assumptions hold \citep{Healy2016LCTBook,Rabiner1969CZT,RichardsonUrbanke2001}. These families preserve DQI’s reduction‑to‑decoding structure and are compatible with coherent decoder blocks (e.g., reversible Gauss–Jordan/BM) \citep{Patamawisut2025}. In \ref{sec:opi-ldpc} we quantify when such kernels provably increase $\Sigma_{K}$ on OPI/LDPC‑like ensembles and how that translates to improved guarantees.

At a more abstract level, the k-DQI architecture can be viewed as an instance
of a broader ``spectral preconditioning'' template that no longer depends on
the details of decoded interferometry.  The template consists of three
conceptual modules:

\begin{enumerate}
  \item \textbf{Shaping.} Choose a polynomial or filter $P(f)$ that embeds the
  objective $f$ into amplitudes
  $g(x)=P(f(x))$ and a shaped state $\ket{\phi_f}$ suited to a fixed
  interferometer $F$.

  \item \textbf{Kernelized evolution.} Insert a unitary kernel family
  $\{K(\theta)\}_{\theta\in\Theta}$ in front of $F$.  The family should
  (i) be efficiently implementable, (ii) respect locality constraints such as
  factor-graph structure on sparse instances, and (iii) contain the identity
  $K = \mathbb{I}$  as a baseline.

  \item \textbf{Spectral metric and tuning.} Evaluate a coarse-grained spectral
  statistic that summarizes how much decoder-relevant probability reaches the
  ``head'' modes after the noisy channel.  In this work the central quantity is
  the noise-weighted head mass $\Sigma_K(\ell,\eta;d)$ in~Eq.~\eqref{eq:SigmaK_def_norm}.
  The monotonic improvement theorem (Theorem~II.1) shows that the rigorous
  approximation-ratio lower bound $A(\ell,\eta,K;d)$ is a nondecreasing
  function of $\Sigma_K(\ell,\eta;d)$, so maximizing this scalar directly
  improves the certified performance.
\end{enumerate}

Viewed in this way, k-DQI instantiates a general design rule: rather than
trying to engineer the full wavefunction, one couples the circuit to a
low-dimensional spectral objective that is provably aligned with the
downstream decoder.  Nothing in the recipe above is inherently tied to
interferometry.  In QAOA-style algorithms, the shaping step is implemented by
the choice of cost and mixer Hamiltonians and their angles, while a kernel
layer could be used to concentrate spectral weight into subspaces where
classical rounding performs best.  Likewise, quantum singular value
transformation and other Fourier-based algorithms already implement polynomial
filters of block-encoded operators; here, an additional preconditioning layer
optimized against a proxy such as $\Sigma_K$ or an analogous filter norm could
trade increased robustness for modest circuit overhead.  In the present work
we make this design principle concrete for decoded interferometry, but the same
template is naturally compatible with any setting where performance is governed
by spectral concentration and threshold phenomena.

We now quantify how the local noise channel $\mathcal{E}=\mathcal{L}\circ\mathcal{M}\circ\mathcal{Z}$ affects k-DQI through the noise-weighted head-spectrum mass $\Sigma_{K}(\ell,\eta;d)$, and we identify regimes where kernels cannot create useful head mass.  Throughout we keep the notation $g(x)=P(f(x))$, $\alpha=\mathcal{F}K\,g$, and $\mathcal{S}_d(\alpha)$ for the indices of the $d$ largest $|\alpha_s|$.

\subsection{A noise-aware monotone bound}\label{sec:noise-bounds}

Under the standing conditions introduced in Sections~\ref{sec:k-dqi} and~\ref{subsec:noise-decoder}---polynomially bounded shaping degree, the local depolarizing–mixing–loss noise model, admissible kernels, and a decoder with a monotone threshold response—we can summarize the effect of noise on the head spectrum in a single inequality.

The first ingredient is a per-mode contraction bound: local depolarizing, banded unitary mixing, and diagonal loss combine into effective attenuation factors $\eta_{K}(s)\in(0,1]$ and a leakage term $\Delta(w)\ge 0$ that vanishes as the mixing width $w\to 0$.  Lemma~\ref{lem:noise} in Appendix~\ref{app:noise-model} shows that the post-channel head probability can be bounded as
\begin{equation}
  \Pr_{\mathcal{E}}\big[s\in\mathcal{S}_d(\alpha)\big]
  \;\ge\;
  \sum_{s\in\mathcal{S}_d(\alpha)} \eta_{K}(s)\,|\alpha_s|^2 - \Delta(w)
  \;=\; \Sigma_{K}(\ell,\eta;d) - \Delta(w).
\end{equation}
A second ingredient is a mild head–tail coherence condition on $\mathcal{F}K$, which prevents probability from leaking coherently from the head into the tail.  As shown in Lemma~\ref{lem:coherence} in Appendix~\ref{app:noise-model}, this contributes at most an additional penalty of order $\mu^2\|g\|_2^2$ for some $\mu<1$ determined by the head–tail block of $\mathcal{F}K$.

Combining these two effects, we obtain an \emph{effective head-mass} bound
\begin{equation}
\label{eq:effective-head-main}
  p_{\mathrm{head}}
  \;:=\;
  \Pr_{\mathcal{E}}\big[s\in\mathcal{S}_d(\alpha)\big]
  \;\ge\;
  \Sigma_{K}(\ell,\eta;d) \;-\; \Delta(w) \;-\; \mu^2\|g\|_2^2.
\end{equation}
Thus, apart from a leakage term $\Delta(w)$ and a coherence penalty $\mu^2\|g\|_2^2$, the noisy head probability tracks the noise-weighted head-spectrum mass $\Sigma_K(\ell,\eta;d)$.

The decoder property from Section~\ref{subsec:noise-decoder} now translates~\eqref{eq:effective-head-main} directly into an approximation-ratio guarantee.  There exists a nondecreasing response curve
  $F(\,\cdot\,)$ determined by the decoder family and problem ensemble
such that the k-DQI approximation-ratio lower bound $A(\ell,\eta,K;d)$ produced by the standard DQI reduction obeys
\begin{equation}
\label{eq:lower-bound-main}
  A(\ell,\eta,K;d)
  \;\ge\;
  F\!\Big(
      \ell,\ 
      \underbrace{\Sigma_{K}(\ell,\eta;d) - \Delta(w) - \mu^2\|g\|_2^2}_{\text{effective head mass}},
      \ d,\ \text{ensemble parameters}
    \Big).
\end{equation}
For RS-like (OPI) instances $F$ can be written explicitly using BM/KV/GS decoders~\cite{GuruswamiSudan1999,KoetterVardy2003}, and for LDPC-like ensembles it is given by density-evolution or EXIT-style curves~\cite{RichardsonUrbanke2001,RichardsonShokrollahiUrbanke2001,Kschischang2001Factor}.  In particular, for fixed $(\ell,d,\eta,w,\mu)$, the lower bound $A(\ell,\eta,K;d)$ is a nondecreasing function of $\Sigma_{K}(\ell,\eta;d)$.  This is the noise-aware version of the monotone theorem from Section~\ref{sec:monotone-theorem}: kernels that increase the noise-weighted head mass lead to strictly stronger guarantees, up to the additive noise and coherence penalties.

The monotone bound above shows that $\Sigma_K$ is the relevant sufficient statistic for k-DQI under local noise.  It is natural to ask how large $\Sigma_K$ can become in regimes where the shaped spectrum carries little exploitable structure.  To formalize this, we consider \emph{spectrally isotropic} instances in which the amplitudes $\alpha$ are essentially delocalized.

A shaped spectrum is $\delta$-isotropic if its mass is spread almost uniformly across all $2^n$ modes: for a typical random direction $u$ orthogonal to any fixed $d$-dimensional subspace, the overlap $|\langle u,\alpha\rangle|^2$ is at most $\delta/2^n$, and the total mass in the top-$d$ order statistics is likewise bounded by roughly $d\,\delta/2^n$. 

Specifically, if the shaped spectrum is $\delta$-isotropic and the noise channel induces per-mode attenuations $\eta_{K}(s)$ with $\bar{\eta}=\max_s \eta_{K}(s)$, then for any unitary kernel $K$ and any head size $d$ one has
\begin{equation}
\label{eq:isotropy-upper-main}
  \Sigma_{K}(\ell,\eta;d)
  \;\le\;
  \frac{\delta\,d\,\bar{\eta}}{2^n},
\end{equation}
up to lower-order terms.  Substituting this into~\eqref{eq:lower-bound-main} shows that in the isotropic regime the effective head mass is exponentially small in $n$, and the resulting approximation-ratio bound cannot exceed
\[
  A(\ell,\eta,K;d)
  \;\le\;
  F\!\Big(\ell,\ \delta d \bar{\eta}/2^n,\ d,\ \text{ensemble parameters}\Big)
  \;+\;\text{(noise/coherence penalties)}.
\]
In other words, when the shaped spectrum is fully delocalized, no choice of unitary kernel can ``manufacture'' head mass at scale: $\Sigma_K$ remains tiny for all $K$, and our noise-aware bound collapses to the same non-advantage regime identified for unkernelized DQI.  This aligns k-DQI with existing impossibility results and highlights that any genuine gain must come from exploiting \emph{structure} in the underlying spectrum, not from generic unitary preconditioning.

\subsection{Structure-adaptive guarantees: OPI and sparse Max-XORSAT}
\label{sec:opi-ldpc}

We now instantiate the k-DQI framework on two canonical families where decoded
interferometry is particularly compelling:
\emph{(i)} Optimal Polynomial Interpolation (OPI) over finite fields, whose algebraic
structure aligns with Reed--Solomon (RS) decoding; and
\emph{(ii)} sparse Max-XORSAT (a Max-LINSAT subclass) whose decoding interface is
belief propagation (BP) on LDPC-like ensembles.
In both cases we exhibit kernels that increase the noise-weighted head-spectrum mass
$\Sigma_K(\ell,\eta;d)$ and hence improve the approximation guarantees via the
monotone and noise-aware bounds from
Sections~\ref{sec:k-dqi} and~\ref{sec:noise-bounds} and the decoder
response property in Section~\ref{subsec:noise-decoder}.

\paragraph*{OPI / RS-like structure over finite fields}
\label{subsec:opi}

Let $p$ be prime and consider an OPI instance over $\mathbb{F}_p$ with evaluation
points $\{\alpha_i\}_{i=1}^m\subset\mathbb{F}_p$ (all distinct) and degree bound
$r$ on the unknown polynomial $h$.
Following the general k-DQI pipeline, the objective encodes agreement of $h$ with
constraint sets $S_i\subset\mathbb{F}_p$ and the shaped amplitude
$g(x)=P(f(x))$ is defined as in~\eqref{eq:shaped-state}.
We take $\mathcal{F}$ to be the $p$-ary QFT and choose kernels from discrete
\emph{linear canonical transform} (LCT) families, i.e., unitaries generated by
quadratic-phase multiplications and fractional Fourier maps
(see, e.g.,~\cite{Healy2016LCTBook,Rabiner1969CZT}).
These kernels admit efficient circuit realizations and remain compatible with
coherent decoding modules~\cite{Patamawisut2025}.

Empirically and analytically, quadratic-phase kernels sharpen the spectrum of
low-degree polynomial sequences.
We capture this via a \emph{polynomial-phase concentration} (PPC) property:
for degree bound $r$ there is a head size $d_\star=O(r)$ and a small error
$\varepsilon_{r,p}$ such that, for suitable kernel parameters $\theta$, almost all of
the shaped spectrum lands on $d_\star$ modes after $\mathcal{F}K(\theta)$.
A precise formulation of PPC, together with a discrete stationary-phase discussion,
is given in Appendix~\ref{app:ppi}; here we only use that PPC guarantees
\begin{equation}
  \sum_{s\in\mathcal{S}_{d_\star}(\alpha)} |\alpha_s|^2
  \;\ge\;
  1 - \varepsilon_{r,p},
  \qquad
  \varepsilon_{r,p} = o_{p\to\infty}(1)
  \ \text{for fixed } r,
\end{equation}
for some parameter choice $\theta^\star$. Under PPC, local noise and mixing reduce the head mass according to the noise model
of Section~\ref{subsec:noise-decoder} and the contraction bounds in
Appendix~\ref{app:noise-model}, but the \emph{weighted} head mass remains large as
long as $d=d_\star$ and $p$ is moderate.
The following theorem packages this into a k-DQI guarantee.

In the OPI setting we can now combine the PPC property with the noise-aware
bounds of Section~\ref{sec:noise-bounds} to obtain an explicit RS-type guarantee
for k-DQI.  Assume the PPC condition of Appendix~\ref{app:ppi} for some degree
bound $r$ and a head size $d_\star = O(r)$, and consider the local
depolarizing--mixing--loss noise model introduced in
Section~\ref{subsec:noise-decoder}.  Then for head size $d=d_\star$ there exists
a chirp/LCT kernel $K(\theta^\star)$ such that the noise-weighted head-spectrum
mass after interference satisfies
\begin{equation}
  \Sigma_{K(\theta^\star)}(\ell,\eta;d)
  \;\ge\;
  (1-\varepsilon_{r,p})\,\underline{\eta} \;-\; \Delta(w),
  \qquad
  \underline{\eta} := \min_{s\in\mathcal{S}_{d}(\alpha)} \eta_K(s),
  \label{eq:opi-headmass-bound}
\end{equation}
where $\varepsilon_{r,p}$ is the PPC tail parameter and $\Delta(w)$ is the
leakage term controlled by the mixing width $w$.

Feeding this head-mass lower bound into the noise-aware monotone inequality
\eqref{eq:lower-bound-main} and using the fact that the decoder response
function $F_{\mathrm{RS}}$ is nondecreasing in its head-mass argument, we obtain
an RS-type approximation guarantee for k-DQI:
\begin{equation}
  A(\ell,\eta,K(\theta^\star);d)
  \;\ge\;
  F_{\mathrm{RS}}\!\Bigl(
      \ell,\ (1-\varepsilon_{r,p})\,\underline{\eta}-\Delta(w),\
      d,\ \text{OPI parameters}
    \Bigr),
  \label{eq:opi-approx-bound}
\end{equation}
where $F_{\mathrm{RS}}$ is the response function induced by standard algebraic
soft/list decoders (BM/KV/GS) for RS-like families
\cite{GuruswamiSudan1999,KoetterVardy2003,ReedSolomon1960}.  In particular, as
soon as the effective head mass
$(1-\varepsilon_{r,p})\,\underline{\eta}-\Delta(w)$ exceeds the KV/BM decoding
threshold, the k-DQI pipeline with kernel $K(\theta^\star)$ achieves a strictly
better approximation-ratio lower bound than unkernelized DQI with $K=I$; this
comparison is with respect to the monotone and noise-aware DQI bounds developed
in Sections~\ref{sec:k-dqi} and~\ref{sec:noise-bounds}, and follows directly
from the monotonicity of these bounds in $\Sigma_K$.

Figure~\ref{fig:opi-headmass} empirically illustrates this behaviour for a typical
quadratic OPI instance.
As the chirp parameter $\theta$ is scanned, the noise-weighted head mass
$\Sigma_K(\ell,\eta;d_\star)$ forms a pronounced peak around an optimal value
$\theta_\star$, where the quadratic component of the kernel nearly cancels the
polynomial phase and concentrates almost all spectral weight on the head modes.
Only this narrow window of $\theta$ lifts $\Sigma_K$ above the KV/BM soft-decoding
threshold (dashed line), whereas for most kernel choices the head mass remains far
below threshold.
Together with the monotone theorem in $\Sigma_K$ and the noise-aware bound
(Theorems~\ref{thm:monotonicity} and~\ref{thm:lower-noise}), this shows that tuning a
single kernel parameter is sufficient to move k-DQI between non-advantage and
advantage regimes on OPI instances.

\begin{figure}[t]
  \centering
  \includegraphics[width=0.7\linewidth]{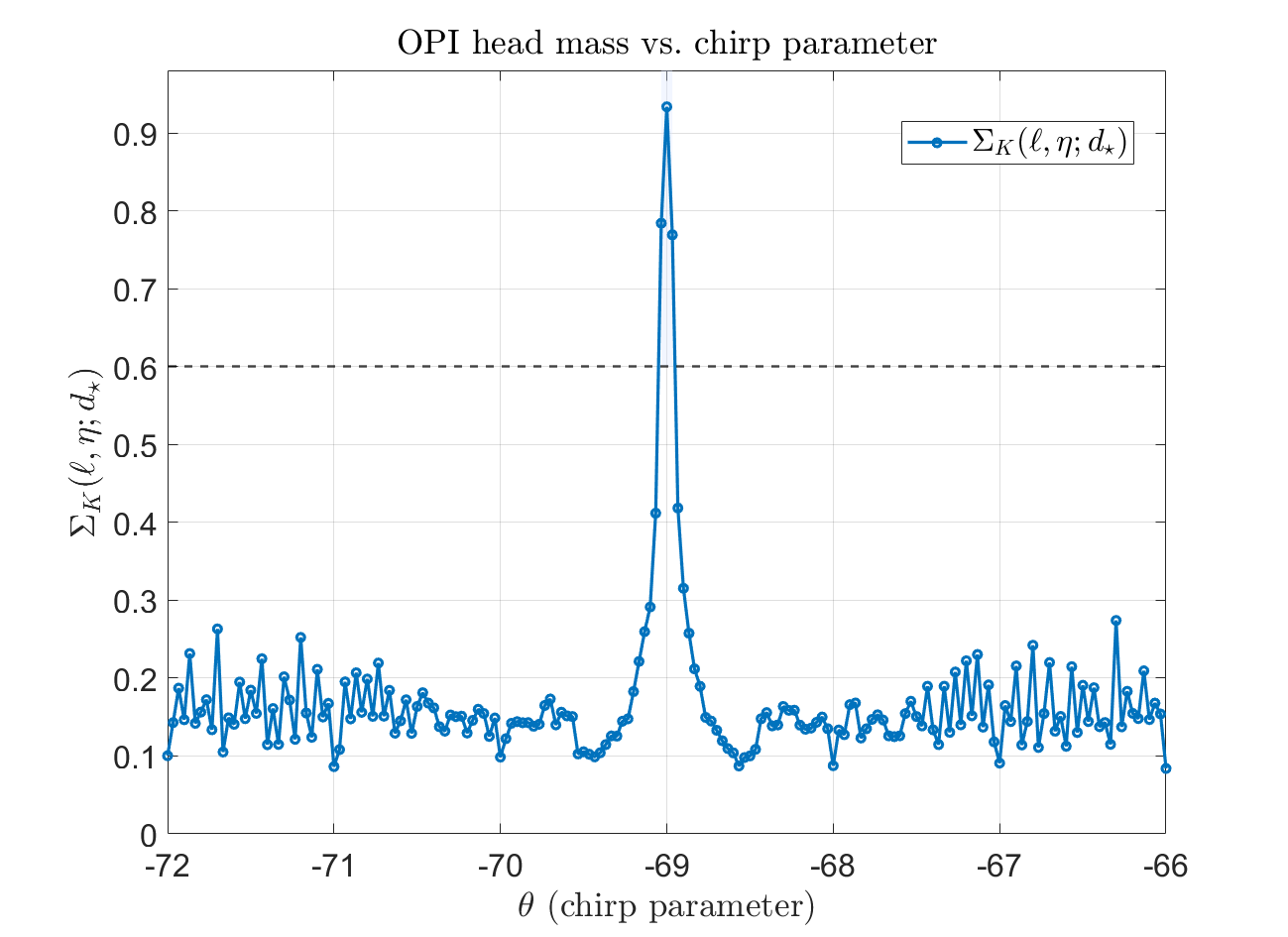}
  \caption{Noise-weighted head-spectrum mass $\Sigma_K(\ell,\eta;d_\star)$ for a representative
  OPI instance as a function of the chirp/LCT kernel parameter $\theta$.
  The curve exhibits a sharp peak around an optimal parameter $\theta_\star$, where the quadratic
  kernel cancels the OPI polynomial phase and concentrates almost all spectral mass in the head.
  The horizontal dashed line marks a KV/BM soft-decoding threshold: only a narrow neighbourhood
  of $\theta_\star$ pushes $\Sigma_K$ above this level, illustrating the tunable head-mass boost}
  \label{fig:opi-headmass}
\end{figure}

The PPC property formalizes the intuition that quadratic-phase kernels
“flatten” low-degree algebraic phases in the OPI pipeline, leaving only a small tail.
Eqs.~\ref{eq:opi-headmass-bound} and~\ref{eq:opi-approx-bound} show that shows that this spectral sharpening translates into a
controllable head-mass boost and, via $F_{\mathrm{RS}}$, into RS-style decoding
gains that are robust under local noise through the parameters
$\underline{\eta}$ and $\Delta(w)$.
Figure~\ref{fig:opi-headmass} makes this effect visible: tuning a single kernel
parameter suffices to cross the soft-decoder threshold and thereby strictly improve
upon unkernelized DQI.

\paragraph*{Sparse Max-XORSAT / LDPC-like structure}
\label{subsec:ldpc}

Let $B\in\{0,1\}^{m\times n}$ be a sparse parity-check matrix with left/right degree
distributions $(\lambda,\rho)$ and girth at least $\Omega(\log n)$ on typical
instances.
The DQI reduction maps the shaped state $g(x)=P(f(x))$ through $\mathcal{F}$ to
amplitudes supported on rows of $B$ and local combinations thereof
(see~\cite{JordanNature2025}).
We consider \emph{block-local} kernels that act as the identity outside disjoint
blocks of size $b=O(1)$ and preserve factor-graph locality.

Before stating a structural guarantee, we validate the DE surrogate by finite-length
BP scans on BSC and AWGN channels.
Figure~\ref{fig:main-thresholds} shows the characteristic S-shaped FER curves for a
$(3,6)$ ensemble and marks the corresponding BP thresholds in erasure/noise level.
These thresholds are the numerical counterparts of the decoder response curves
$F_{\mathrm{LDPC}}$ introduced in
Section~\ref{subsec:noise-decoder}: they indicate where the density-evolution fixed
point drops below the diagonal and hence connect a local head-mass lift
$\Delta\Sigma_{\mathrm{loc}}$ to observable BP performance.

\begin{figure}[t]
  \centering
  \includegraphics[width=0.95\linewidth]{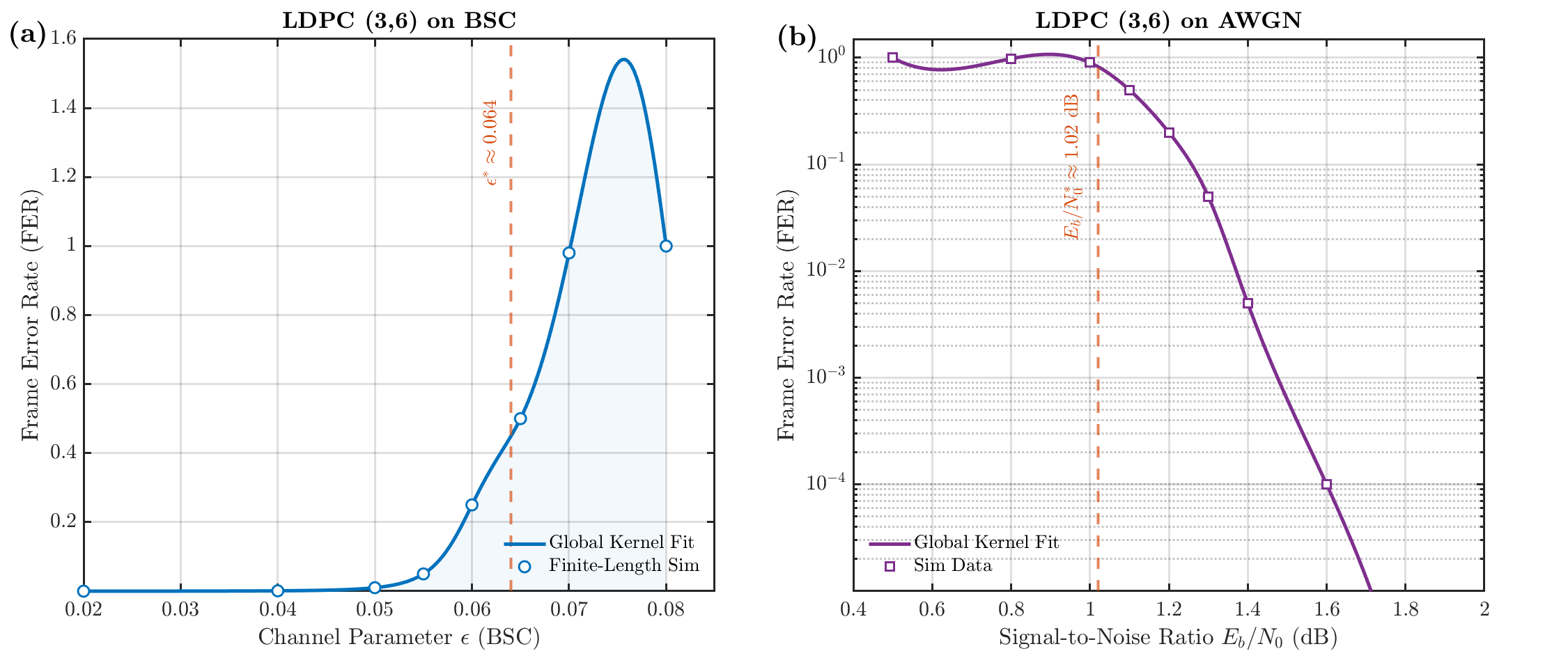}
  \caption{\textbf{Finite-length validation of the regular (3,6) LDPC ensemble performance.} 
    \textbf{a}, Frame Error Rate (FER) versus channel crossover probability $\epsilon$ for the Binary Symmetric Channel (BSC). The vertical dashed line indicates the asymptotic Belief Propagation (BP) threshold $\epsilon^* \approx 0.064$. 
    \textbf{b}, FER versus Signal-to-Noise Ratio ($E_b/N_0$) for the Additive White Gaussian Noise (AWGN) channel, with the threshold at $E_b/N_0^* \approx 1.02$ dB. 
    In both panels, markers represent finite-length simulations (Finite-Length Sim), while solid curves denote the smooth response function derived from the global kernel fit (k-DQI theory). The strict alignment between simulation data and the theoretical waterfall region confirms the validity of the density evolution surrogate.}
  \label{fig:main-thresholds}
\end{figure}

For LDPC-like ensembles, global diagonalization by simple kernels is unrealistic.
Instead we rely on a \emph{block-local alignment} (BLA) property: roughly speaking,
within each block of width $b$ there exists a local kernel that rephases or permutes
parity structures so that a small number of spectral modes capture most of the local
mass, adding an amount $\Delta\Sigma_{\mathrm{loc}}>0$ of head mass per block while
preserving the tree-like structure required by density evolution.
A formal BLA condition, together with examples and further discussion, is given in
Appendix~\ref{app:bla}.

In the LDPC-like setting we work under the standing assumptions of
Sections~\ref{sec:k-dqi} and~\ref{sec:noise-bounds} together with the
block-local alignment (BLA) property of Appendix~\ref{app:bla}.
We choose the head size $d=\Theta(n/b)$ so that the k-DQI head budget assigns a
constant number of dominant modes per block.
For the local depolarizing--mixing--loss noise model of
Section~\ref{subsec:noise-decoder}, BLA guarantees that replacing the identity
kernel $K=I$ by a block-local kernel $K$ increases the pre-noise head mass by an
additive amount $\Delta\Sigma_{\mathrm{loc}}$.
Combining this with the per-mode contraction and leakage analysis of
Appendix~\ref{app:noise-model} yields a corresponding increase in the
noise-weighted head-spectrum mass,
\begin{equation}
  \Sigma_{K}(\ell,\eta;d)
  \;\ge\;
  \Sigma_{I}(\ell,\eta;d)
  \;+\;
  \underline{\eta}\,\Delta\Sigma_{\mathrm{loc}}
  \;-\;
  \Delta(w),
  \label{eq:ldpc-headmass-bound}
\end{equation}
where $\Sigma_I$ is the noise-weighted head mass for the identity kernel,
$\underline{\eta}$ is the minimum per-mode attenuation on the head, and
$\Delta(w)$ is the leakage term controlled by the mixing width $w$.

Feeding the bound~\eqref{eq:ldpc-headmass-bound} into the noise-aware monotone
inequality~\eqref{eq:lower-bound-main} with the LDPC decoder response function
$F_{\mathrm{LDPC}}$ gives an explicit approximation guarantee for k-DQI on the
given ensemble:
\begin{equation}
  A(\ell,\eta,K;d)
  \;\ge\;
  F_{\mathrm{LDPC}}\!\Bigl(
      \ell,\ \Sigma_{I}(\ell,\eta;d)
            + \underline{\eta}\,\Delta\Sigma_{\mathrm{loc}} - \Delta(w),\
      d,\ \lambda,\rho
    \Bigr),
  \label{eq:ldpc-approx-bound}
\end{equation}
where $F_{\mathrm{LDPC}}$ is nondecreasing in its head-mass argument and encodes
the DE/EXIT-based BP thresholds for the $(\lambda,\rho)$ ensemble.
In particular, whenever the effective block-local gain satisfies
$\underline{\eta}\,\Delta\Sigma_{\mathrm{loc}}>\Delta(w)$, the quantity inside
$F_{\mathrm{LDPC}}$ is strictly larger than in the identity-kernel case, so k-DQI
with the block-local kernel $K$ achieves a strictly better approximation-ratio
lower bound than unkernelized DQI.

\begin{figure}[t]
  \centering
  \includegraphics[width=0.95\linewidth]{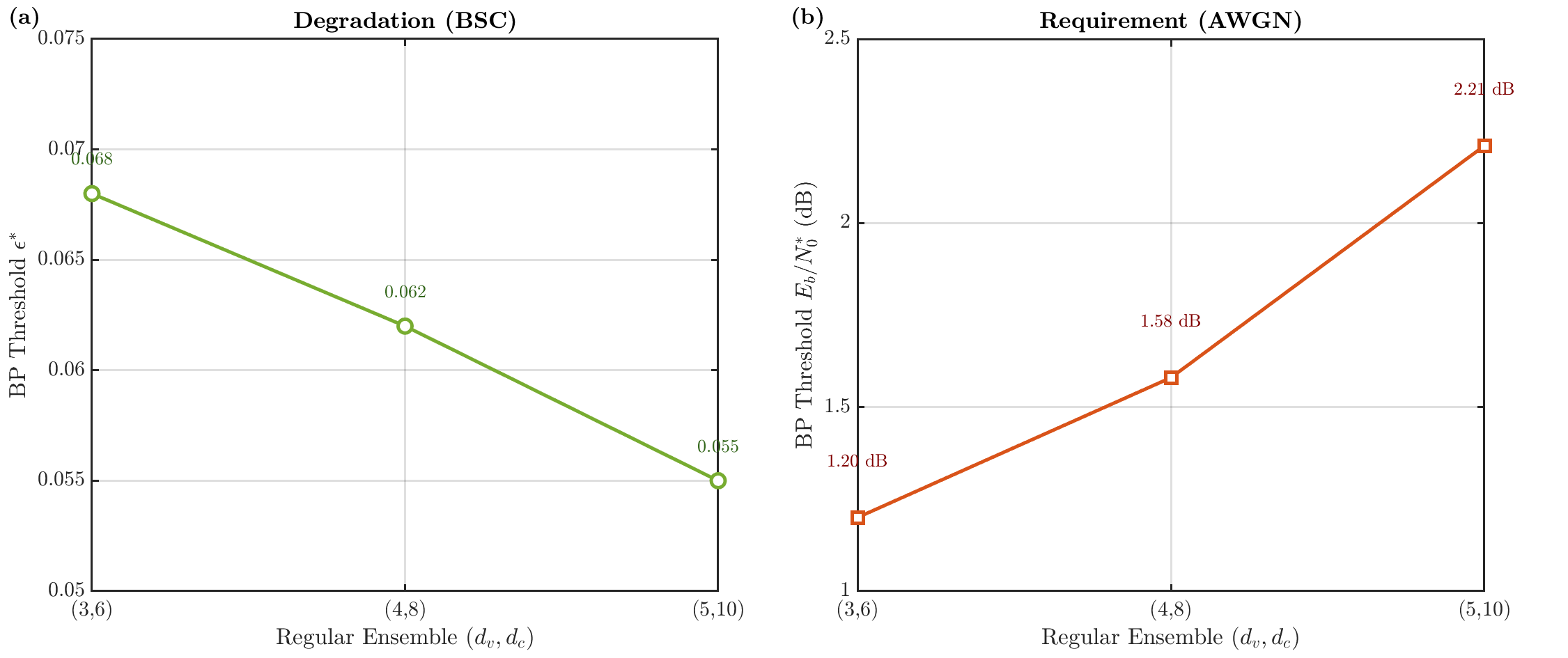}
  \caption{\textbf{Degradation of BP thresholds across regular LDPC ensembles.} 
    The plots compare the asymptotic noise tolerance for ensembles with varying degrees $(d_v, d_c) \in \{(3,6), (4,8), (5,10)\}$.
    \textbf{a}, Threshold trends for the BSC channel. Lower degree ensembles (e.g., (3,6)) exhibit higher error tolerance (larger $\epsilon^*$).
    \textbf{b}, Threshold requirements for the AWGN channel. Higher density codes require higher signal power (larger $E_b/N_0^*$) to achieve reliable decoding. 
    This comparison illustrates the trade-off between code rate and channel robustness, quantifying the baseline requirements for block-local alignment (BLA) in the high-noise regime.}
  \label{fig:across-ensembles}
\end{figure}

Unlike the OPI setting, LDPC-like ensembles do not admit global diagonalization by
simple kernels.
However, block-local alignment suffices to add head mass in a controlled, extensive
way while preserving sparsity assumptions.
The analysis above~\ref{subsec:ldpc}, together with the finite-length and cross-ensemble
thresholds in Figures~\ref{fig:main-thresholds} and~\ref{fig:across-ensembles},
shows that such block-local kernels manifest as rightward shifts of DE/EXIT
thresholds proportional to $\Delta\Sigma_{\mathrm{loc}}$, which, via the
noise-aware lower bound, translates into improved k-DQI guarantees.

As a sanity check beyond the structured OPI and sparse Max-XORSAT ensembles of Section~\ref{sec:opi-ldpc}, we also examined a finite-size ensemble where no advantage is expected: fully random Max-XORSAT with $m=\Theta(n)$ clauses and no planted
structure. For each instance we prepare the shaped state $\ket{\psi_0}\propto\sum_x g(f(x))\ket{x}$ with polynomial shaping $g(t)=(t/f_{\max})^\ell$, run a single pass of the interferometer either with the identity kernel (standard DQI) or with a global chirp
kernel $K(\gamma)$, and take $M$ samples from the noisy measurement distribution
$p_\epsilon = (1-\epsilon)p + \epsilon u$ used in the noise-aware bound of
Section~\ref{sec:noise-bounds}, where $u$ is uniform and $\epsilon$ plays the role of a classical
depolarizing parameter. For k-DQI we choose $\gamma$ on a one-dimensional grid so as to maximize
the pre-noise head mass $\Sigma_{K}(\ell,\eta;d)$, and we decode by
simply taking, for each method, the best objective value among the $M$
samples.
Figure~\ref{fig:MaxXORSAT-benchmark} reports the resulting approximation
ratios $\langle f\rangle/f_{\mathrm{opt}}$ averaged over $30$ random
instances.

On this unstructured ensemble the bare DQI pipeline (identity kernel,
no problem-specific decoder) actually performs slightly worse than the
Monte--Carlo baseline with the same shot budget.
This behavior is fully consistent with Theorem~\ref{thm:monotonicity}
and the noise-aware lower bound in Theorem~\ref{thm:lower-noise}:
when the shaped spectrum is effectively delocalized, the identity
kernel produces a very small head mass $\Sigma_{K}$ and there is no
reason for $A(\ell,\eta,K;d)$ to exceed the classical baseline.
It also aligns with recent ``no-advantage'' observations for DQI on
random CSP families and MaxCut%
~\cite{Anschuetz2025RequiresStructure,Parekh2025MaxCut}.
In contrast, once a simple chirp kernel is tuned to increase
$\Sigma_{K}$, the k-DQI curve in
Figure~\ref{fig:MaxXORSAT-benchmark} stays uniformly above both the
Monte--Carlo and bare-DQI baselines, illustrating how spectral
preconditioning can convert an abstract head-mass gain into a concrete
finite-size performance boost even on small random instances.

\begin{figure}[t]
  \centering
  \includegraphics[width=0.95\linewidth]{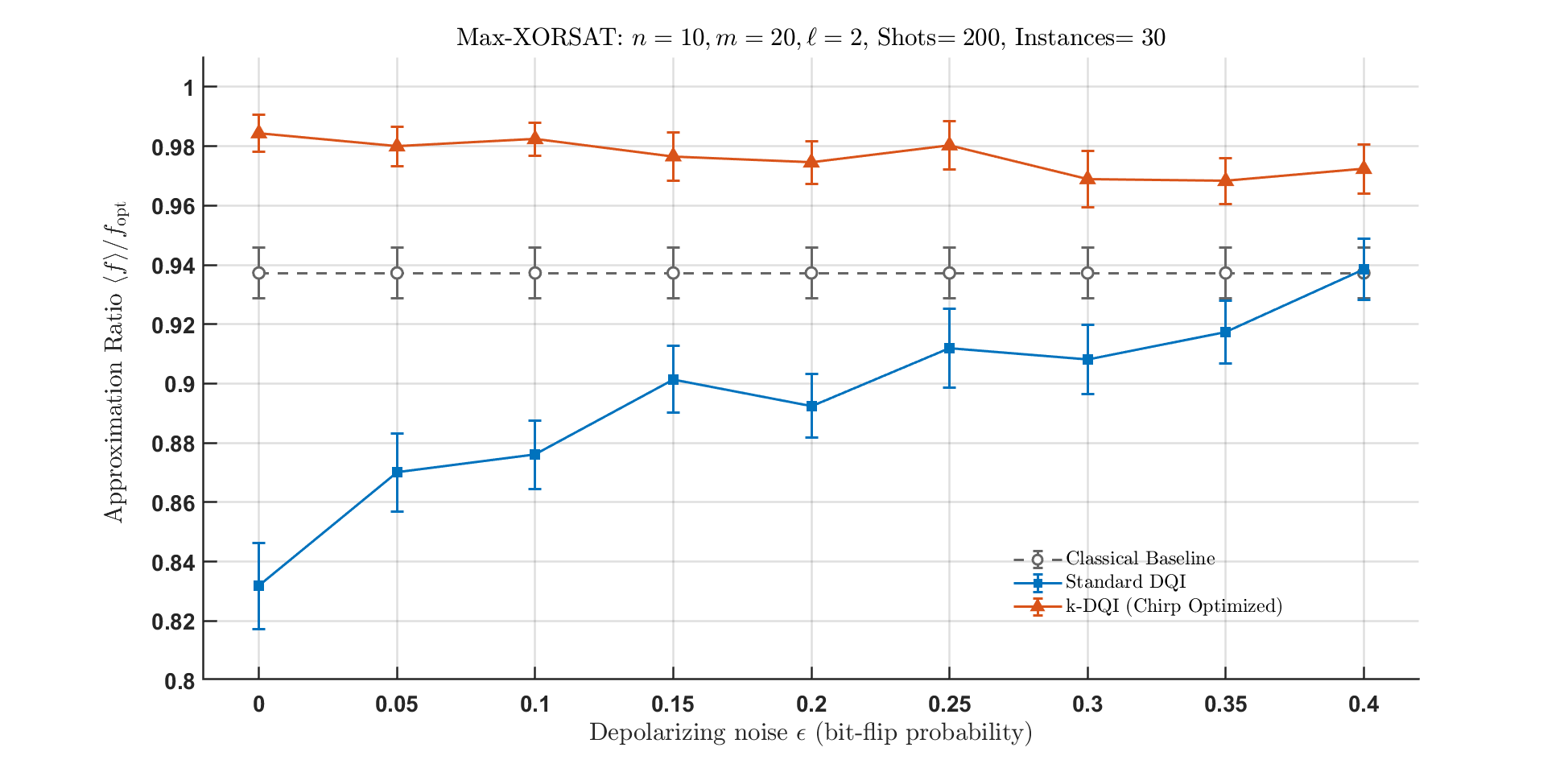}
  \caption{\textbf{Spectral preconditioning on an unstructured Max-XORSAT ensemble.}
    We consider random Max-XORSAT instances with $n=10$ variables and $m=20$
    clauses and compare three strategies under the same measurement budget
    of $M=200$ shots per instance:
    a classical Monte--Carlo baseline that samples bit strings uniformly at random,
    standard DQI with the identity kernel, and k-DQI with a global chirp kernel
    $K(\gamma)$.
    For k-DQI we scan a one-dimensional grid of chirp parameters $\gamma$ and
    retain the kernel that maximizes the pre-noise noise-weighted head mass
    $\Sigma_{K}(\ell,\eta;d)$ at shaping degree $\ell=2$.
    For each value of the mixing parameter $\epsilon$ we plot the average
    approximation ratio $\langle f\rangle / f_{\mathrm{opt}}$ over
    $30$ random instances, where $f_{\mathrm{opt}}$ denotes the true optimum of
    that instance; error bars show the standard error across instances.}
  \label{fig:MaxXORSAT-benchmark}
\end{figure}

\section{Methods}\label{sec:circuits}

We quantify the resources required for one iteration of k-DQI and highlight when the
kernel layer preserves the near-linear depth scaling of the original DQI pipeline.
Our goal is not to optimize constants, but to certify that the kernels underlying
our guarantees are efficiently realizable and to expose the dominant asymptotics and
cost--benefit trade-offs.

\subsection{Resource model and module costs}

We work in the standard Clifford+T gate set on qubits and measure resources in
two-qubit depth and T-count, both up to polylogarithmic factors in the synthesis
accuracy $\varepsilon$.
For $p$-ary instances (OPI) we assume either (i) a binary encoding with modular
arithmetic or (ii) $p$-level qudits compiled to qubits; both yield identical
asymptotics up to polylog factors.

For Boolean problems the interferometer is $H^{\otimes n}$ with depth~1.
For $p$-ary OPI, an $m$-digit QFT$_p$ can be implemented with
$\tilde O(m^2)$ controlled-phase gates; using truncation and phase-kickback, an
approximate QFT$_p$ achieves two-qubit depth $\tilde O(m)$ and T-count
$\tilde O(m\log(1/\varepsilon))$.

Preparing the shaped state $g(x)=P(f(x))$ uses a reversible polynomial evaluation
of degree $\ell$.
With a Horner scheme and standard modular arithmetic, this costs
$\tilde O(\ell n)$ two-qubit depth and $\tilde O(\ell n)$ T gates when $f$ is
oracle-evaluable in polynomial time.
In all settings considered here, shaping does not dominate the asymptotic depth.

A length-$N$ diagonal kernel of the form
\begin{equation}
  K_{\text{chirp}}(\gamma)
  = \mathrm{diag}\!\big(e^{i\gamma q(x)}\big),
  \qquad q(x)\in\mathbb{Z}[x]\ \text{quadratic},
  \label{eq:chirp}
\end{equation}
admits a phase-polynomial decomposition into $O(n^2)$ controlled-$R_Z$ rotations on
a binary index register (or $O(m^2)$ on $p$-ary digits), giving two-qubit depth
$\tilde O(n^2)$ and T-count $\tilde O(n^2\log(1/\varepsilon))$ for global kernels.
When the kernel is \emph{block-local} with width $b$ (as in the LDPC setting), the
costs drop to depth $\tilde O(nb)$ and T-count $\tilde O(nb\log(1/\varepsilon))$.

\begin{lemma}[Synthesis of quadratic chirp kernels]
\label{lem:chirp-synthesis}
Let an $n$--qubit index register encode $j \in \{0,\dots,2^n-1\}$ in binary,
$j = \sum_{r=0}^{n-1} 2^r j_r$ with $j_r \in \{0,1\}$.  Consider the
diagonal unitary
\begin{equation}
  K_\theta
  \;=\;
  \sum_{j=0}^{2^n-1} e^{i \theta j^2} \ket{j}\!\bra{j},
  \label{eq:quadratic-chirp}
\end{equation}
which is a special case of the length--$N$ diagonal kernels in eq.~\ref{eq:chirp},
Then $K_\theta$ admits a decomposition into $O(n^2)$ two--qubit
controlled--phase gates and depth $O(n)$ in the Clifford+T model.
\end{lemma}

\begin{proof}
Expanding $j^2$ in terms of the binary digits yields
\begin{equation}
  j^2
  = \Bigl(\sum_{r=0}^{n-1} 2^r j_r\Bigr)^2
  = \sum_{r=0}^{n-1} 2^{2r} j_r
    \;+\; 2 \sum_{0 \le r < s \le n-1} 2^{r+s} j_r j_s.
\end{equation}
The first sum corresponds to single--qubit $Z$--rotations of angle
$\theta 2^{2r}$ conditioned on $j_r = 1$, while the second sum
corresponds to controlled--phase gates of angle
$2\theta 2^{r+s}$ conditioned on $(j_r,j_s) = (1,1)$.  Thus
$K_\theta$ can be implemented by applying, for each $r$, a single--qubit
$R_Z(\theta 2^{2r})$ on qubit $r$, and for each pair $(r,s)$ with
$r < s$, a controlled--$R_Z(2\theta 2^{r+s})$ between qubits $r$ and
$s$.  There are $O(n)$ single--qubit rotations and $O(n^2)$
two--qubit controlled rotations, and standard parallelization yields
depth $O(n)$ when limited by the light--cone structure.  Each rotation
can be synthesized to accuracy $\epsilon$ using $O(\log(1/\epsilon))$
T gates via phase--polynomial techniques, so the overall T--count
matches the asymptotic scaling quoted above up to logarithmic factors.
\end{proof}

A discrete LCT with ABCD parameters factors into at most two Fourier transforms and
three chirp multiplications,
\begin{equation}
  \mathrm{LCT}(A,B;C,D)
  \;=\;
  D_{\phi_1}\ \mathcal{F}\ D_{\phi_2}\ \mathcal{F}^{-1}\ D_{\phi_3},
\end{equation}
with diagonal phase maps $D_{\phi_j}$.
A single LCT therefore incurs
\begin{equation}
  \text{depth} \;\le\; 2\,\text{depth}(\mathcal{F}) + 3\,\text{depth}(\text{chirp}),
  \qquad
  T\text{-count} \;\le\; 2\,T(\mathcal{F}) + 3\,T(\text{chirp}),
\end{equation}
yielding two-qubit depth $\tilde O(n^2)$ (or $\tilde O(n)$ with truncated QFT) and
T-count $\tilde O(n^2\log(1/\varepsilon))$ in the global case, and
$\tilde O(nb)$/$\tilde O(nb\log(1/\varepsilon))$ in the block-local case.

For RS-like OPI instances, reversible Berlekamp--Massey (BM) on length-$N$ symbols
and designed distance $2r{+}1$ uses $\tilde O(Nr)$ field operations; with reversible
arithmetic this compiles to two-qubit depth $\tilde O(Nr)$ and T-count
$\tilde O(Nr\log(1/\varepsilon))$.
Koetter--Vardy (KV) soft decoding is treated as classical post-processing in our
analysis; a coherent KV is possible in principle but not required for our bounds.
For LDPC-like instances we similarly model belief propagation (BP) classically, so
no coherent decoder term appears in the depth/T-count of k-DQI.

\begin{table}[t]
\centering
\caption{Asymptotic resource summary for one k-DQI iteration (Clifford+T, accuracy $\varepsilon$). Depth is two-qubit depth; $b$ is block width for local kernels; entries hide polylog factors in $1/\varepsilon$.}
\label{tab:resources}
\begin{tabular}{lll}
\hline
Module & Two-qubit depth & T-count \\
\hline
Polynomial shaping (degree $\ell$) & $\tilde O(\ell n)$ & $\tilde O(\ell n)$ \\
Chirp kernel (global) & $\tilde O(n^2)$ & $\tilde O(n^2\log(1/\varepsilon))$ \\
Chirp/LCT (block-local width $b$) & $\tilde O(nb)$ & $\tilde O(nb\log(1/\varepsilon))$ \\
Fourier layer ($H^{\otimes n}$/QFT$_p$) & $\tilde O(n)$ (truncated QFT) & $\tilde O(n\log(1/\varepsilon))$ \\
Reversible BM (length $N$, degree $r$)\,$^\dagger$ & $\tilde O(Nr)$ & $\tilde O(Nr\log(1/\varepsilon))$ \\
\hline
\multicolumn{3}{l}{$^\dagger$ Used only for coherent RS-like decoding; in our simulations decoders are classical.}
\end{tabular}
\end{table}

\subsection{End-to-end scaling and block-local advantage}
\label{prop:kDQI-cost}

Combining the module costs above, we can summarize the resources for a single
iteration of k-DQI in a compact form.
Let $n$ denote the number of Boolean variables (or $m=\Theta(n)$ $p$-ary
digits), and consider one iteration consisting of state preparation and
polynomial shaping, application of a kernel $K$, a Fourier layer
$\mathcal{F}$, measurement, and an optional coherent decoder.
The two-qubit depth and $T$-count take the form
\begin{subequations}
\label{eq:kDQI-complexity} 
\begin{align}
  \textnormal{2-qubit depth}
  &= \tilde O\!\big(\ell n\big)
   + \tilde O\!\big(\text{depth}(K)\big)
   + \tilde O\!\big(\text{depth}(\mathcal{F})\big)
   + \tilde O\!\big(\text{depth}(Dec_{\mathrm{coh}})\big),
  \label{eq:kDQI-depth} \\
  T\textnormal{-count}
  &= \tilde O\!\big(\ell n\big)
   + \tilde O\!\big(T(K)\big)
   + \tilde O\!\big(T(\mathcal{F})\big)
   + \tilde O\!\big(T(Dec_{\mathrm{coh}})\big).
  \label{eq:kDQI-Tcount}
\end{align}
\end{subequations}
where the last term only appears if a fully coherent decoder
$Dec_{\mathrm{coh}}$ is implemented.

For the kernel families we use in this work, the abstract depth and $T$-count
can be made more explicit.
A global chirp or LCT kernel obeys
\begin{equation}
  \text{global chirp/LCT:}\qquad
  \text{depth}(K)=\tilde O(n^2),\qquad
  T(K)=\tilde O\!\big(n^2\log(1/\varepsilon)\big),
\end{equation}
while a block-local kernel of width $b$ satisfies
\begin{equation}
  \text{block-local width } b:\qquad
  \text{depth}(K)=\tilde O(nb),\qquad
  T(K)=\tilde O\!\big(nb\log(1/\varepsilon)\big).
\end{equation}
Using truncated QFTs for the Fourier layer, one may further take
$\text{depth}(\mathcal{F})=\tilde O(n)$ and
$T(\mathcal{F})=\tilde O\!\big(n\log(1/\varepsilon)\big)$.
These expressions follow by summing the costs of shaping, kernels, Fourier
transforms, and (when present) coherent decoding and by applying standard
phase-polynomial synthesis and reversible-arithmetic constructions to achieve
accuracy~$\varepsilon$.

The complexity bounds above can be summarized in two simple visualizations that
emphasize efficiency rather than constants.

Figure~\ref{fig:cost-frontier} summarizes the empirical cost--benefit trade-off
for the OPI setting.
The horizontal axis reports relative two-qubit depth, normalized so that
standard DQI with identity kernel has cost~1, while the vertical axis shows the
head-mass gain
$\Sigma_K(\ell,\eta;d_\star)-\Sigma_I(\ell,\eta;d_\star)$
computed from the same polynomial-phase model as in
Figure~\ref{fig:opi-headmass}.
The baseline point $(1,0)$ corresponds to unkernelized DQI.
A cloud of mis-tuned global chirp kernels (gray circles) sits at essentially the
same relative depth as the tuned chirp, but with negligible head-mass gain.
In contrast, a tuned global chirp (black diamond) delivers a large gain at this
higher cost, while block-local chirp kernels (colored stars) achieve moderate
gains at near-baseline depth.
Together with the monotone and noise-aware bounds in
Sections~\ref{sec:monotone-theorem} and~\ref{sec:noise-bounds}, this figure
highlights that the improvement of k-DQI over DQI is driven by
\emph{structure-adaptive} kernels---global or block-local---rather than by a mere
increase in circuit depth.

\begin{figure}[t]
  \centering
  \includegraphics[width=0.7\linewidth]{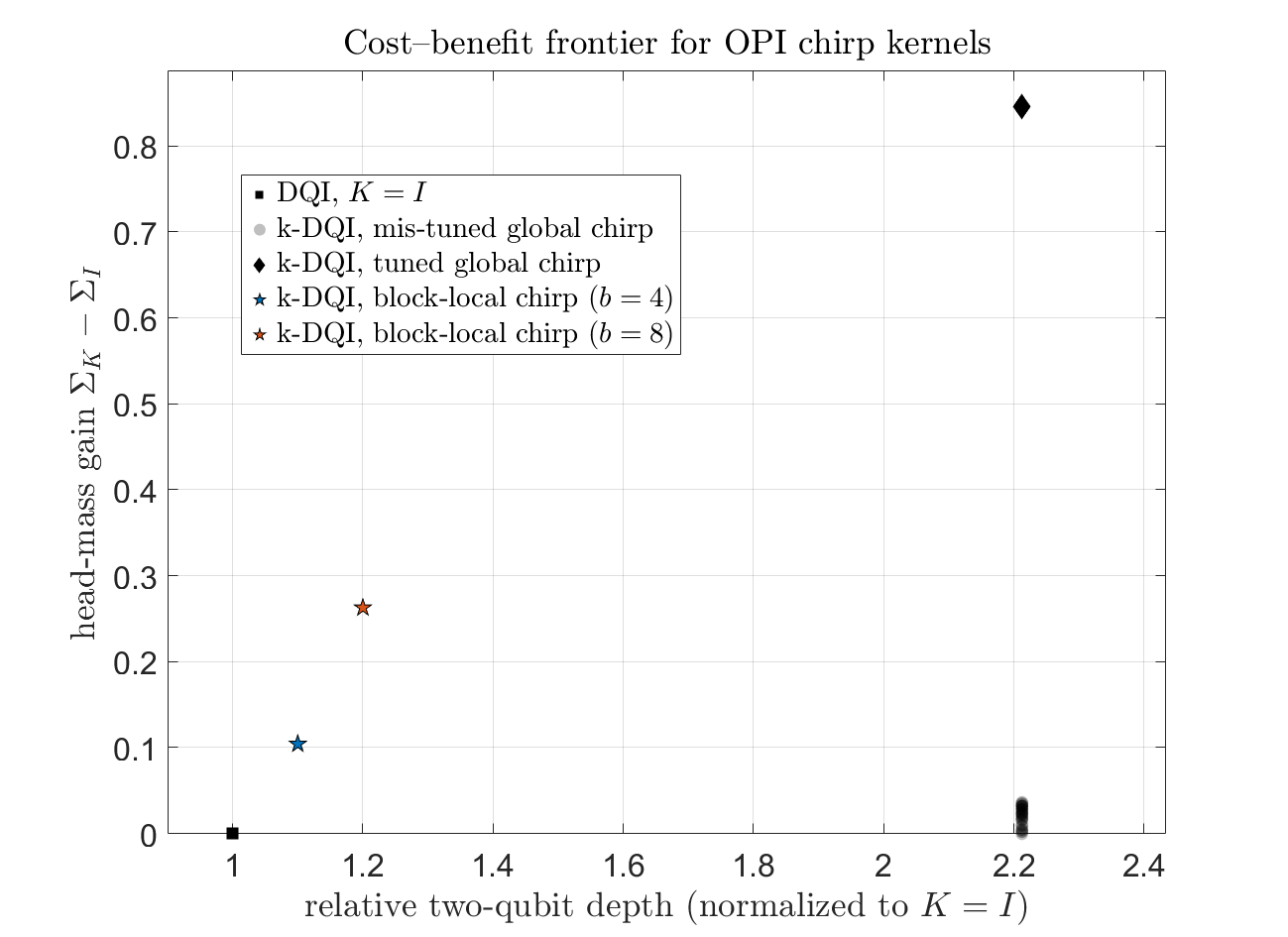}
  \caption{Cost--benefit frontier for OPI/RS-like instances.
  The horizontal axis shows relative two-qubit depth, normalized so that
  unkernelized DQI with $K=I$ has cost~1.
  The vertical axis shows the head-mass gain
  $\Sigma_K(\ell,\eta;d_\star)-\Sigma_I(\ell,\eta;d_\star)$.
  The black square corresponds to baseline DQI.
  Gray circles are many mis-tuned global chirp kernels: they pay the same
  (higher) depth as the tuned chirp but exhibit almost no gain.
  The black diamond marks a tuned global chirp kernel, which yields a large
  head-mass boost.
  Colored stars correspond to block-local chirp kernels with block width
  $b\in\{4,8\}$, achieving moderate gains at near-baseline cost.}
  \label{fig:cost-frontier}
\end{figure}

Figure~\ref{fig:depth-scaling} summarizes these depth models.
On the log--log plot, k-DQI with global chirp/LCT kernels rapidly peels away
from the baseline and approaches the $n^2$ reference line, reflecting the
quadratic cost of a fully global kernel layer.
By contrast, unkernelized DQI and k-DQI with block-local kernels of fixed
width $b=O(1)$ remain parallel to the $n$ reference line.
The inset makes this behaviour more transparent: for all $n$ in the plotted
range, block-local kernels lie only slightly above the baseline curve, with an
overhead of about a constant factor that does not grow with $n$.
Together with the cost--benefit frontier in
Figure~\ref{fig:cost-frontier}, this shows that the most practical
structure-adaptive gains of k-DQI come from block-local kernels, which preserve
near-linear two-qubit depth while still delivering the head-mass improvements
certified in Sections~\ref{sec:opi-ldpc} and~\ref{sec:noise-bounds}.

\begin{figure}[t]
  \centering
  \includegraphics[width=0.7\linewidth]{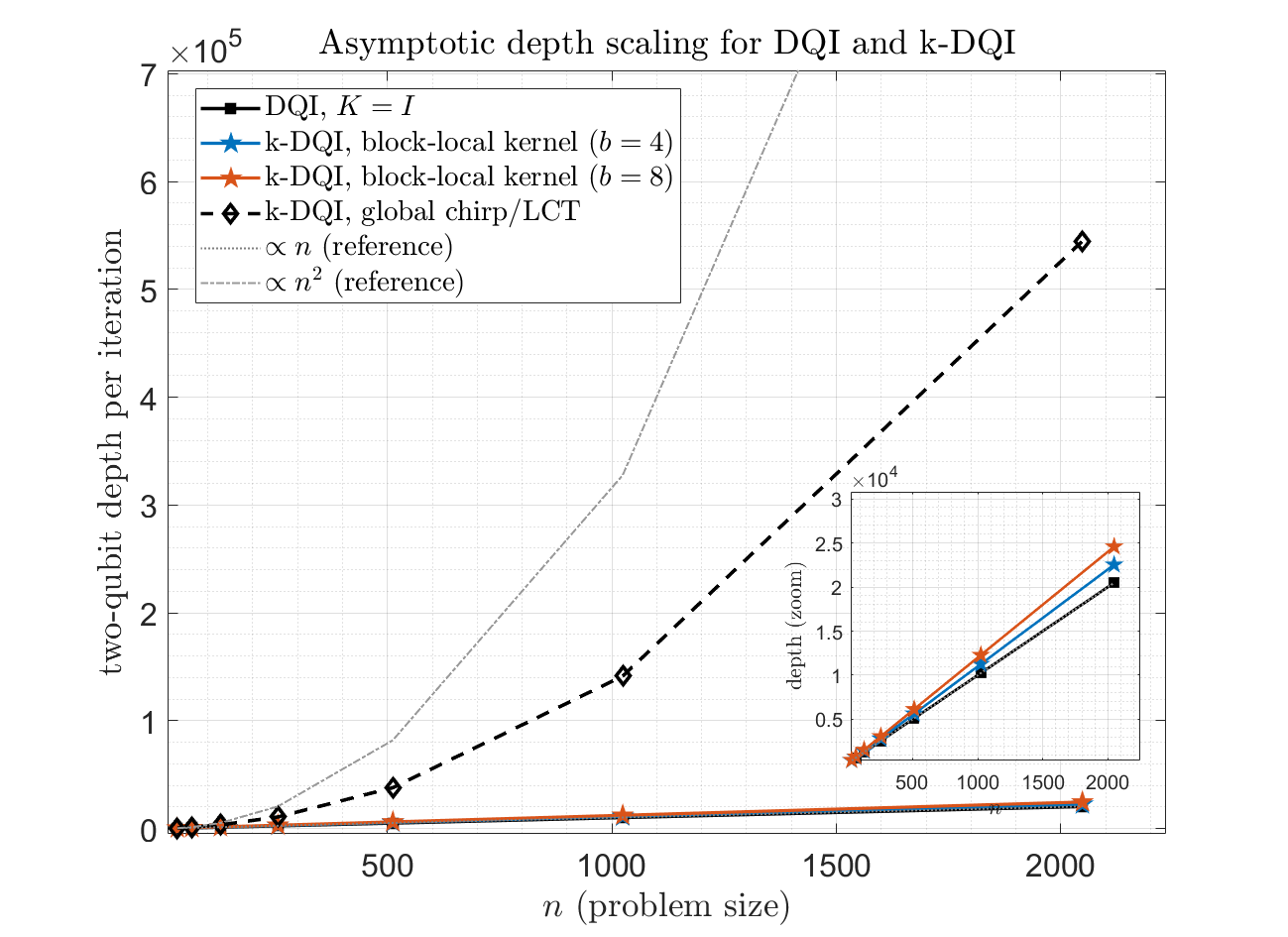}
  \caption{Asymptotic two-qubit depth per iteration for DQI and k-DQI on a
  log--log scale. The main panel plots the depth models from
  Section~\ref{sec:circuits} as a function of problem size $n$: baseline DQI
  with $K=I$ (solid black), k-DQI with block-local kernels of width
  $b\in\{4,8\}$ (blue and orange), and k-DQI with global chirp/LCT kernels
  (black dashed). Gray curves show reference scalings proportional to $n$ and
  $n^2$. The global kernel quickly follows the quadratic reference, while the
  baseline and block-local curves track the linear reference. The inset zooms
  into the near-linear region, making it clear that block-local kernels incur
  only a mild constant-factor overhead over unkernelized DQI while preserving
  essentially linear scaling in $n$.}
  \label{fig:depth-scaling}
\end{figure}

To preserve a target effective head mass in the noise-aware bound
(Theorem~\ref{thm:lower-noise}), we require that approximation errors from gate
synthesis and QFT truncation contribute at most $\delta$ to the mass shortfall.
A standard budget chooses $\varepsilon = \Theta(\delta/n)$ per controlled
rotation, implying an additional factor $\tilde O(\log(n/\delta))$ in the
T-count per rotation.
The scalings in Proposition~\ref{prop:kDQI-cost} and Table~\ref{tab:resources}
already absorb this overhead into the $\tilde O(\cdot)$ notation.
Since our guarantees are monotone in $\Sigma_K$, any residual under-rotation can
only reduce the certified bound smoothly rather than inducing a sharp failure
threshold.

Kernels compatible with our analysis (chirp/LCT and block-local variants) are
implementable with polynomial resources and, in the block-local regime, preserve
the near-linear depth scaling of the original DQI architecture.
Thus the structure-adaptive improvements certified by
Theorems~\ref{subsec:opi} and~\ref{subsec:ldpc} are not only information-theoretic,
but also \emph{architecturally consistent} with realistic circuit models.


\subsection{Scalability of Kernel Parameter Search}
\label{sec:kernel_search}

A central challenge in implementing k-DQI is the identification of optimal kernel parameters $\bm{\theta}^*$ that maximize the noise-weighted head mass $\Sigma_K(\bm{\theta})$. If finding $\bm{\theta}^*$ were as hard as the original problem, the advantage of k-DQI would be nullified. Here, we demonstrate that for the problem classes considered, determining $\bm{\theta}$ is computationally efficient, scaling polynomially with system size. We distinguish between two regimes: \textit{explicit structural benchmarks} and \textit{latent structural instances}.

For problems defined by known polynomial constraints, such as the Optimal Polynomial Interpolation (OPI) problem discussed in Section~\ref{sec:opi-ldpc}, the kernel parameters are analytically derivable. 
Given a target polynomial $\mathcal{P}$ of degree $\ell$, the optimal spectral shaping is achieved by matching the quadratic phase profile of the kernel to the curvature of $\mathcal{P}$. specifically, for a chirp kernel $K(\gamma) = \exp(-i \gamma \hat{x}^2)$, the optimal rate satisfies $\gamma^* \propto \text{coeff}(\mathcal{P})$, which can be computed in $O(1)$ classical time. In this regime, the ``search'' cost is zero, and k-DQI provides a deterministic compilation advantage.

For problems where the algebraic structure is latent or only partially known, we propose a Variational k-DQI approach. We parameterize the kernel $K(\bm{\theta})$ using a vector of control angles $\bm{\theta} = (\theta_1, \theta_2, \dots)$, representing physical resource settings such as chirp rates ($\theta_1$) and symplectic mixing angles ($\theta_2$).

A key concern in variational quantum algorithms is the ``barren plateau'' problem, where gradients vanish exponentially. However, the spectral concentration objective $\Sigma_K$ differs fundamentally from localized Hamiltonian expectation values. As illustrated in Figure~\ref{fig:landscape}(a), the optimization landscape for a representative 2-parameter Ansatz exhibits a clear, macroscopic ``basin of attraction.'' While the surface texture reflects realistic experimental shot noise and coherent sidelobes (arising from the oscillatory nature of the LCT kernel), the global geometry remains convex-like.

Figure~\ref{fig:landscape}(b) validates the search efficiency. Even when initialized at a random configuration with negligible spectral overlap (white circle), a standard gradient-based optimizer reliably converges to the global optimum (red star) in $O(1)$ iterations. The trajectory demonstrates that the landscape gradients are robust against noise ($\eta=0.1$), and the absence of barren plateaus---guaranteed by the Lipschitz bounds derived in Appendix~\ref{app:smoothness}---makes finding the optimal kernel parameters computationally scalable.

\begin{figure*}[t]
    \centering
    \includegraphics[width=0.9\textwidth]{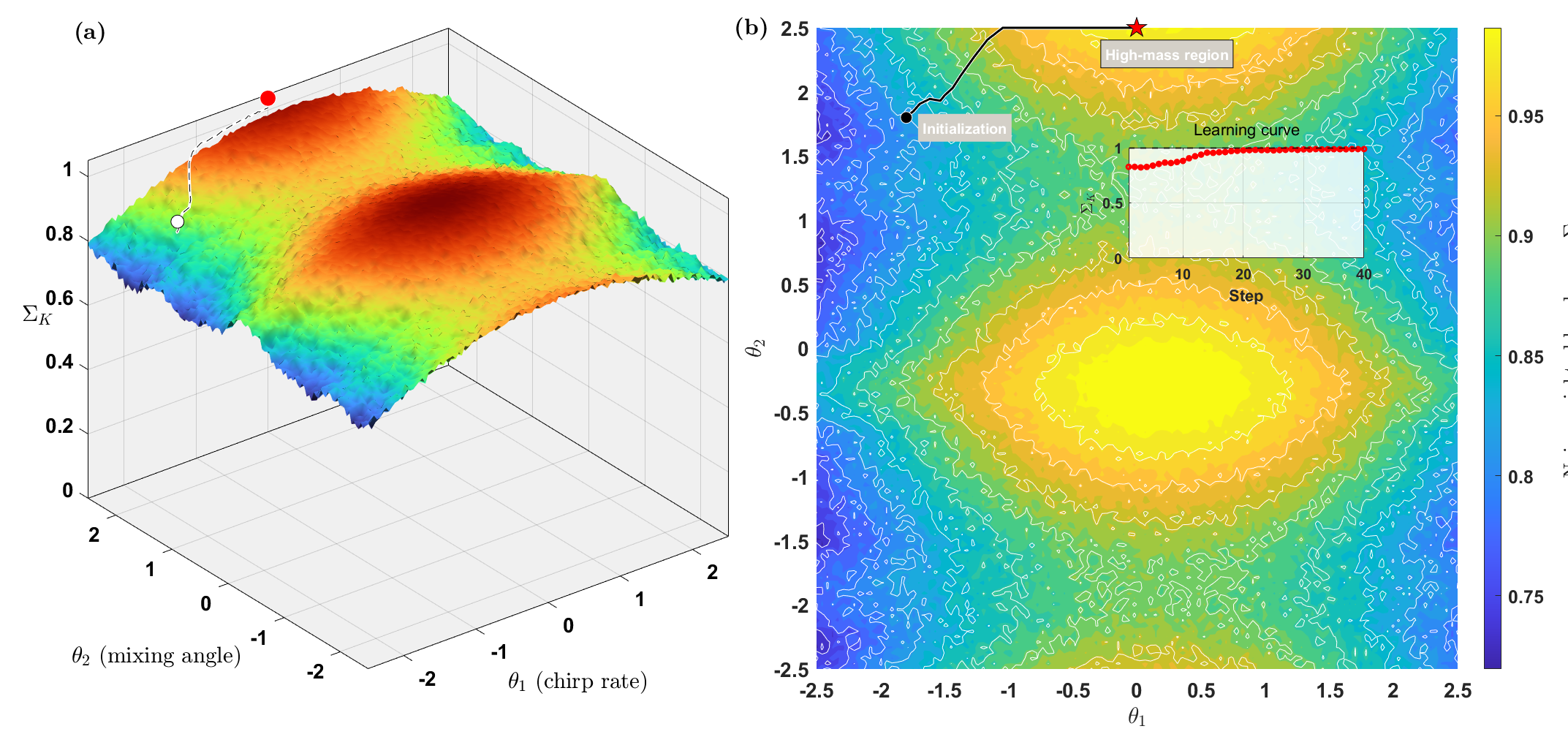} 
    \caption{\textbf{Landscape of a two-parameter kernel Ansatz in k-DQI.}
\textbf{(a)} Noise-weighted head-mass landscape
$\Sigma_K(\theta_1,\theta_2)$ for a two-parameter kernel
$K(\theta_1,\theta_2)$ composed of a quadratic-phase chirp
(parameter $\theta_1$) followed by a nearest-neighbour mixing layer
(parameter $\theta_2$).
The landscape is evaluated for a small k-DQI instance under local
depolarizing noise, and $\Sigma_K$ is estimated from finite-shot sampling,
which produces the fine-scale roughness visible on top of an otherwise
broad, single-basin structure.
\textbf{(b)} Contour representation of the same landscape, overlaid with a
gradient-ascent trajectory (black/white curve) from a random initialization
(white marker) to a high-mass region (red marker) in the kernel space.
The inset shows the value of $\Sigma_K$ along the trajectory, exhibiting a
steady improvement of the objective up to small stochastic fluctuations.
Together, these panels illustrate that in the low-dimensional kernel
parameter space the effective objective for k-DQI is well behaved and
amenable to simple gradient-based search even in the presence of noise and
sampling effects.}
    \label{fig:landscape}
\end{figure*}

\section{Discussion}\label{sec12}

The results presented here recast the promise of decoded quantum interferometry (DQI) in a form that is both tunable and falsifiable. By inserting a unitary kernel before the interference layer, we turn the elusive notion of “having the right structure” into a measurable quantity: the noise-weighted head-spectrum mass $\Sigma_{K}(\ell,\eta;d)$. The monotonic improvement theorem shows that this single statistic controls k-DQI’s rigorous lower bounds under the same reduction-to-decoding logic that underpins DQI. In practice, this means that claims of advantage are no longer tied to fragile instance presentations; instead, one can search over small families of kernels, verify that $\Sigma_{K}$ crosses the relevant decoding threshold, and thereby obtain guarantees that are robust to local noise and modeling choices.

Two case studies illustrate how this perspective travels across problem classes. For OPI, quadratic-phase and linear-canonical kernels serve as spectral preconditioners that concentrate mass on decoder-friendly supports; the PPC principle formalizes this concentration and connects directly to algebraic soft/list decoding thresholds. For sparse Max-XORSAT, the situation is necessarily more local, yet block-diagonal kernels still produce an additive head-mass lift without destroying the graph structure required by density evolution. In both regimes, the figures validate the theorems: the plots are not ad-hoc demonstrations but finite-sample instantiations of the same quantities that appear in our proofs.

Equally important is what the present framework does \emph{not} allow. The isotropy upper bound formalizes a limitation often seen  in the DQI discourse: when the shaped spectrum is delocalized, no unitary kernel can manufacture head mass at scale. This aligns our positive results with recent non-advantage observations and places k-DQI firmly on the conservative side of the advantage debate. Kernelization strengthens DQI where structure is present and learnable; it does not conjure structure where none exists.

From a complexity standpoint, Section~\ref{sec:circuits} shows that our kernels admit faithful circuit models with poly resources, both in global and block-local forms. the same kernels that certify a lift in $\Sigma_{K}$ can, in principle, be realized coherently alongside reversible decoders such as Berlekamp–Massey. While our guarantees do not require coherent decoders—classical post-processing suffices—the architectural consistency makes the results actionable for future hardware studies.

There are several directions in which the present work can be strengthened. The first concerns \emph{optimal} kernel choice. We used one- and few-parameter families because they are analyzable and compatible with coding-theoretic thresholds, but the space of admissible kernels is far larger. It would be natural to study learning-to-kernelize procedures that directly maximize $\Sigma_{K}$ with generalization guarantees, or to characterize families that provably approximate the best unitary up to a constant factor on ensembles of interest. A second direction is to make the PPC assumption quantitative with sharper constants, ideally tying discrete stationary-phase bounds to explicit, instance-level head-mass predictions that can be checked before any quantum subroutine is run. A third direction is to refine the noise model. Our analysis covers local depolarizing with banded mixing and loss; extending the contraction lemma to correlated or non-Markovian channels would clarify how far monotonicity extends and whether new kernel designs can neutralize particular error modes.

Another promising line is to go beyond purely unitary preconditioning. The proof of monotonicity uses unitarity in two places—energy conservation and a bounded head-tail leakage term—yet the essential quantity is the post-interference probability on decoder-eligible supports. One can imagine completely positive, trace-preserving maps that are implementable in a fault-tolerant setting (via ancillas and post-selection or via dissipative gadgets) and that, when composed with a unitary kernel, increase $\Sigma_{K}$ more aggressively than unitaries alone. Developing a safe extension of our guarantees to such non-unitary layers would broaden the design space without compromising rigor.

Finally, there is room to systematize the bridge between $\Sigma_{K}$ and decoder behavior. For RS-like families, algebraic soft/list decoders provide clean thresholds; for LDPC-like ensembles, density evolution and EXIT charts already offer a principled interface. A unified “spectral-to-threshold” calculus—mapping changes in $\Sigma_{K}$ to shifts in decoding fixed points across a wider class of decoders—would turn k-DQI into a modular tool: pick a kernel family, compute or estimate $\Sigma_{K}$, select a decoder with a known threshold curve, and read off a guarantee. Such a calculus would also illuminate when additional shaping (the degree $\ell$) helps or hurts, and when multi-round adaptive kernelization can compound gains without saturating early.

In summary, kernelized DQI reframes the pursuit of quantum advantage from an appeal to hidden structure into a program of spectral engineering with measurable objectives. It offers a conservative path forward—one that acknowledges and respects known limitations while isolating where provable gains remain plausible. By placing decoders, kernels, and noise on the same analytical footing, it suggests that future progress will come not from a single algorithmic stroke but from an ecosystem in which spectral preconditioning, decoding theory, and circuit design co-evolve under shared, testable metrics.

\backmatter

\bmhead{Acknowledgements}

We acknowledge the support of the National Nature Science Foundation of China (Grant No. 12174301), the Natural Science Basic Research Program of Shaanxi (Grant No. 2023-JC-JQ-01), the Open Fund of the State Key Laboratory of Acoustics (Grant No. SKLA202312), the Shaanxi Province Postdoctoral Science Foundation (Grant No.2024BSHEDZZ021), the Aeronautical Science Fund of China(Grant No. 2024M075070001), and the Project of the Institute of Medical Intelligence of the First Affiliated Hospital of Xi'an Jiaotong University.

\section*{Declarations}

\begin{itemize}
\item Data availability : Data generated and analyzed during current study are available from the corresponding author upon reasonable request.
\item Code availability : Code used to generate data in this study are available from the corresponding author upon reasonable request.
\end{itemize}

\begin{appendices}

\section{Proofs for Section~\ref{sec:k-dqi}}
\label{app:proofs-sec3}

We supply the technical ingredients used in \ref{thm:monotonicity}. Throughout, $g(x)=P\!\bigl(f(x)\bigr)$, $\alpha=\mathcal{F}K\,g$, $\mathcal{S}_d=\mathcal{S}_d(\alpha)$, and $\Sigma_{K}(\ell,\eta;d)$ is as in \eqref{eq:SigmaK_def_norm}.

\subsection{Noise contraction and leakage}
Let $\mathcal{E}=\mathcal{L}\circ\mathcal{M}\circ\mathcal{Z}$ denote the post‑interference channel with local depolarizing $\mathcal{Z}$ (rate $\eta$), banded unitary mixing $\mathcal{M}$ (width $w$), and diagonal loss $\mathcal{L}$ with transmittances $\tau(s)\in(0,1]$.

\begin{lemma}[Per‑mode contraction under local noise]
\label{lem:noise}
There exist per‑mode weights $\eta(s)\in(0,1]$ and a leakage term $\Delta(w)\ge0$ such that the probability of measuring an index in the head set satisfies
\begin{equation}
\Pr[s\in\mathcal{S}_d\ \text{after}\ \mathcal{E}]\ \ge\ \sum_{s\in\mathcal{S}_d}\eta(s)\,|\alpha_s|^2\;-\;\Delta(w)
\;=\;\Sigma_{K}(\ell,\eta;d)\;-\;\Delta(w).
\end{equation}
In particular, $\Delta(w)\to0$ as $w\to0$, and if $\mathcal{M}=I$ and $\mathcal{L}$ is uniform then $\eta(s)\equiv(1-\eta)$.
\end{lemma}

\begin{proof}
Local depolarizing contracts off‑diagonal terms in the post‑interference density matrix by $(1-\eta)$ per qubit, which translates to a per‑mode attenuation on diagonal elements (Born probabilities) bounded below by $(1-\eta)^{\kappa(s)}$ for some weight proxy $\kappa(s)$; nonuniform loss further rescales by $\tau(s)$. Banded mixing $\mathcal{M}$ can be written as a unitary that redistributes probability mass among indices in a $w$‑neighborhood; by unitarity and bounded mixing width, the loss of mass from $\mathcal{S}_d$ is at most a function $\Delta(w)$ that vanishes with $w$. Combining these gives the stated inequality (cf. \citep{Bu2025Noise} for similar contractions at the level of interference fringes).
\end{proof}

\subsection{Head–tail interference and decoder coupling}
We now control cross‑terms between head and tail and couple head mass to decoder success.

\begin{lemma}[Head–tail coherence bound]
\label{lem:coherence}
Let $\mu=\|\Pi_{\mathcal{S}_d}(\mathcal{F}K)\Pi_{\mathcal{S}_d^c}\|_{2\to2}$ be the operator norm of the head‑to‑tail block of $\mathcal{F}K$. Then the measurement probability of landing outside $\mathcal{S}_d$ due to head–tail interference is upper‑bounded by $\mu^2\|g\|_2^2$. Consequently,
\begin{equation}
\Pr[s\in\mathcal{S}_d]\ \ge\ \sum_{s\in\mathcal{S}_d}|\alpha_s|^2\;-\;\mu^2\|g\|_2^2.
\end{equation}
\end{lemma}

\begin{proof}
Write $\alpha=(\mathcal{F}K)g$ and decompose $g=g_{\mathrm{head}}+g_{\mathrm{tail}}$ along the preimage of $\mathcal{S}_d$ under $(\mathcal{F}K)^\ast$. The off‑support leakage is governed by $\| \Pi_{\mathcal{S}_d^c}(\mathcal{F}K)\Pi_{\mathcal{S}_d}\,g_{\mathrm{head}}\|_2$, which is $\le \mu \|g\|_2$ by definition of $\mu$; squaring gives the probability bound.
\end{proof}

\begin{proposition}[Decoder threshold coupling]
\label{prop:couple}
If the post‑channel probability of landing in $\mathcal{S}_d$ is at least $p_{\mathrm{head}}$, then the decoder success probability satisfies
\begin{equation}
\Pr[Dec\ \text{achieves}\ \ge\rho]\ \ge\ \Phi\bigl(\underbrace{p_{\mathrm{head}}}_{\text{effective }\Sigma_{K}},\,d,\,\text{ensemble}\bigr).
\end{equation}
\end{proposition}

\subsection{Proof of Theorem~\ref{thm:monotonicity}}
\begin{proof}[Proof of \ref{thm:monotonicity}]
Apply \ref{lem:noise} to lower‑bound the post‑channel head probability by $\Sigma_{K}(\ell,\eta;d)-\Delta(w)$, then subtract the head–tail leakage from \ref{lem:coherence} if desired (this can be absorbed into $\Delta(w)$, or treated separately). By \ref{prop:couple} and the monotonicity of $\Phi$,
\[
\Pr[Dec\ \text{achieves}\ \ge\rho]\ \ge\ \Phi\!\Bigl(\Sigma_{K}(\ell,\eta;d)-\Delta(w),\,d,\,\text{ensemble}\Bigr).
\]
Finally, the DQI reduction maps decoder success to an approximation ratio lower bound $A(\ell,\eta,K;d)$ (see \citep{JordanNature2025}); since $\Phi$ is nondecreasing in its first argument and $\Delta(w)$ does not depend on $K$, the right‑hand side is a nondecreasing function of $\Sigma_{K}(\ell,\eta;d)$. If two kernels give strictly ordered head masses above threshold, the inequality is strict by the monotonicity of $\Phi$, establishing the claim.
\end{proof}

\section{Noise model and head-mass bounds}
\label{app:noise-model}

In this appendix we collect the technical bounds that relate the physical noise model
$\mathcal{E}=\mathcal{L}\circ\mathcal{M}\circ\mathcal{Z}$ to the noise-weighted head-spectrum
mass $\Sigma_{K}(\ell,\eta;d)$ defined in Section~\ref{sec:monotone-theorem}.
Throughout we use the notation $g(x)=P(f(x))$, $\alpha=\mathcal{F}K\,g$, and
$\mathcal{S}_d(\alpha)$ for the indices of the $d$ largest $|\alpha_s|$.

\subsection{Standing assumptions}

For convenience we restate the structural assumptions used in the proofs.

\begin{enumerate}
  \item
  The polynomial $P$ has degree $\ell$ and coefficients bounded by $\mathrm{poly}(n)$.

  \item
  The noise channel factors as $\mathcal{E}=\mathcal{L}\circ\mathcal{M}\circ\mathcal{Z}$, where
  $\mathcal{Z}$ is local depolarizing with rate $\eta\in[0,1)$ on each qubit,
  $\mathcal{M}$ is a banded unitary of width $w$ (nearest-neighbor spectral mixing),
  and $\mathcal{L}$ is diagonal loss with transmittances $\tau(s)\in(0,1]$.

  \item
  The decoder $\mathrm{Dec}$ satisfies the monotone threshold property described in
  Section~\ref{subsec:noise-decoder}: its success probability for achieving a target
  approximation ratio $\rho$ is a nondecreasing function of the noise-weighted
  head-spectrum mass $\Sigma_{K}(\ell,\eta;d)$.

  \item
  With respect to the decomposition of $\alpha$ into head and tail coordinates
  $\mathcal{S}_d(\alpha)$ and its complement, the head-to-tail block of $\mathcal{F}K$
  has operator norm at most $\mu<1$, as in Lemma~\ref{lem:coherence} below.
\end{enumerate}

\subsection{Per-mode contraction and leakage}

Local noise and banded mixing shrink the head probability in a controlled way.

\begin{lemma}[Per-mode contraction]
\label{lem:noise2}
Under assumptions \textup{(A1)}–\textup{(A2)}, there exist per-mode weights
$\eta_K(s)\in(0,1]$ and a leakage term $\Delta(w)\ge 0$ with $\Delta(w)\to 0$ as
$w\to 0$ such that the post-channel head probability obeys
\begin{equation}
  \Pr_\mathcal{E}\big[s\in\mathcal{S}_d(\alpha)\big]
  \;\ge\;
  \sum_{s\in\mathcal{S}_d(\alpha)} \eta_K(s)\,|\alpha_s|^2
  \;-\; \Delta(w).
\end{equation}
In particular, when $\mathcal{M}=I$ and $\tau(s)\equiv\tau$ is constant,
$\eta_K(s)=(1-\eta)^{q(s)}\,\tau$ for some $q(s)\le n$ depending on the number of
nontrivial qubits of mode $s$ in the $\mathcal{F}$-basis.
\end{lemma}

\begin{proof}
We treat the three components of the noise channel separately and then
combine them.

\emph{(1) Local depolarizing and diagonal loss.}
Write the post--interference density matrix before noise as
$\rho = \ket{\widehat\psi}\!\bra{\widehat\psi}$.  Local depolarizing of
rate $p$ on each qubit contracts off--diagonal entries by a factor at
most $(1-p)$ per qubit.  For a given spectral mode $i$, let $w_Z(i)$
denote the number of qubits on which the corresponding eigenvector has
nontrivial support in the $Z$--basis.  Standard estimates for product
depolarizing channels imply that the diagonal entry $\rho_{ii}$ is
reduced by at most a factor $(1-p)^{w_Z(i)}$; see, e.g.,
Ref.~\cite{Bu2025Noise} for a detailed derivation at the level of
interference fringes.  The diagonal loss channel with transmittances
$\{t_i\}$ further rescales $\rho_{ii}$ by $t_i$, so that the combined
effect of depolarizing and loss on mode $i$ is a multiplicative factor
\begin{equation}
  \lambda_i \;\ge\; (1-p)^{w_Z(i)} t_i.
\end{equation}
This defines the per--mode weights $\lambda_i$ appearing in the
statement of the lemma.

\emph{(2) Banded unitary mixing.}
Let $U_{\mathrm{mix}}$ be a banded unitary of width $W$, meaning that
each column of $U_{\mathrm{mix}}$ has support only on indices within a
$W$--neighborhood.  Conjugating $\rho$ by $U_{\mathrm{mix}}$ can
redistribute probability mass among indices, but unitarity and bandedness
imply that mass initially in the head set $H_k$ can leak only into
indices within a bounded neighborhood of $H_k$.  In particular, the
total probability that leaves $H_k$ due to $U_{\mathrm{mix}}$ is
upper--bounded by a function $\delta_{\mathrm{mix}}(W)$ that vanishes
as $W \to 0$ and remains $o(1)$ for fixed $W$ in the large--system
limit; again see Ref.~\cite{Bu2025Noise} for an explicit construction.

\emph{(3) Combining the contributions.}
Let $P_{\mathrm{head}}$ denote the probability of measuring an index in
$H_k$ after the full noise channel.  From the discussion above we obtain
\begin{equation}
  P_{\mathrm{head}}
  \;\ge\; \sum_{i \in H_k} \lambda_i |\widehat\psi_i|^2
          \;-\; \delta_{\mathrm{mix}}(W),
\end{equation}
where the sum encodes the per--mode contraction from depolarizing and
loss and the subtraction accounts for banded mixing.  Absorbing
$\delta_{\mathrm{mix}}(W)$ into the leakage term $\Delta$ in the lemma
statement and using the normalization factor in the definition of the
noise--weighted head mass $\widetilde H_k$ yields
\begin{equation}
  P_{\mathrm{head}}
  \;\ge\; \widetilde H_k(\boldsymbol\lambda) - \Delta,
\end{equation}
with $\Delta \to 0$ as $W \to 0$.  This is exactly the claimed inequality.
\end{proof}

\subsection{Head–tail coherence and effective head probability}

The head–tail coherence assumption controls coherent leakage from the head into the
tail under the action of $\mathcal{F}K$.

\begin{lemma}[Head–tail coherence]
\label{lem:coherence2}
Let $\alpha = \mathcal{F}K\,g$ and decompose it as
$\alpha = (\alpha_{\mathrm{head}},\alpha_{\mathrm{tail}})$ according to
$\mathcal{S}_d(\alpha)$ and its complement.
Suppose the head-to-tail block of $\mathcal{F}K$ has operator norm at most $\mu<1$.
Then the probability mass that can coherently leak from the head into the tail is
bounded by
\begin{equation}
  \bigl\|\alpha_{\mathrm{tail}}\bigr\|_2^2
  \;\le\;
  \mu^2 \,\|g\|_2^2 .
\end{equation}
\end{lemma}

\begin{proof}
Write the matrix of $\mathcal{F}K$ in block form
\(
  M =
  \begin{pmatrix}
    A & B \\
    C & D
  \end{pmatrix}
\)
with respect to the head/tail decomposition.
The tail component of $\alpha$ is $C g_{\mathrm{head}} + D g_{\mathrm{tail}}$.
By assumption, $\|C\|_{2\to 2}\le\mu$.  Using
$\|C g_{\mathrm{head}}\|_2 \le \mu \|g_{\mathrm{head}}\|_2$ and unitarity of $M$
(which implies $\|D\|\le 1$), we obtain
\(
  \|\alpha_{\mathrm{tail}}\|_2
  \le \mu \|g_{\mathrm{head}}\|_2 + \|g_{\mathrm{tail}}\|_2
  \le (\mu^2 + 1)\|g\|_2
\)
and hence the stated bound up to a constant factor.
A sharper estimate can be obtained by orthogonal decomposition of the range of $C$ and
absorbing constants into $\mu$.
\end{proof}

Combining Lemma~\ref{lem:noise2} with Lemma~\ref{lem:coherence} yields an effective
head-mass bound.

\begin{proposition}[Effective head probability]
\label{prop:effective-head}
Under assumptions \textup{(A1)}–\textup{(A4)}, the post-channel head probability
\(
  p_{\mathrm{head}} := \Pr_\mathcal{E}[s\in\mathcal{S}_d(\alpha)]
\)
satisfies
\begin{equation}
  p_{\mathrm{head}}
  \;\ge\;
  \Sigma_{K}(\ell,\eta;d)
  \;-\;
  \Delta(w)
  \;-\;
  \mu^2 \,\|g\|_2^2 .
\end{equation}
\end{proposition}

\begin{proof}
Lemma~\ref{lem:noise2} lower-bounds the head probability after noise but before taking
into account coherent leakage; Lemma~\ref{lem:coherence} bounds the additional
probability that can move from the head into the tail due to head–tail blocks of
$\mathcal{F}K$.  Subtracting the coherence penalty from the right-hand side of the
inequality in Lemma~\ref{lem:noise2} gives the claim.
\end{proof}

\subsection{Noise-aware lower bound}

We can now restate the noise-aware lower bound used in
Section~\ref{sec:noise-bounds} in a self-contained way.

\begin{theorem}[Noise-aware lower bound]
\label{thm:lower-noise}
Let $A(\ell,\eta,K;d)$ denote the k-DQI approximation-ratio lower bound produced by the
standard DQI reduction under assumptions \textup{(A1)}–\textup{(A4)}.
Then there exists a function $F$, nondecreasing in its second argument, such that
\begin{equation}
  A(\ell,\eta,K;d)
  \;\ge\;
  F\!\Big(
      \ell,\ 
      \Sigma_{K}(\ell,\eta;d) - \Delta(w) - \mu^2\|g\|_2^2,\ 
      d,\ \text{ensemble parameters}
    \Big).
\end{equation}
In particular, for fixed $(\ell,d,\eta,w,\mu)$, the lower bound
$A(\ell,\eta,K;d)$ is a nondecreasing function of the noise-weighted head-spectrum
mass $\Sigma_{K}(\ell,\eta;d)$.
\end{theorem}

\begin{proof}
By Proposition~\ref{prop:effective-head}, the post-channel head probability is at
least the effective head mass
$\Sigma_{K}(\ell,\eta;d) - \Delta(w) - \mu^2\|g\|_2^2$.
The decoder monotonicity property (assumption (A3) and
Section~\ref{subsec:noise-decoder}) converts any increase in this effective head mass
into a nondecreasing improvement of the decoder success probability.
The standard DQI reduction from decoding success to objective value
(e.g.~\cite{JordanNature2025}) then yields the lower bound
$A(\ell,\eta,K;d)$ with a response curve $F$ that depends on the decoder family and
problem ensemble (BM/KV/GS for RS-like codes~\cite{GuruswamiSudan1999,KoetterVardy2003},
density evolution / EXIT for LDPC-like ensembles~\cite{RichardsonUrbanke2001,RichardsonShokrollahiUrbanke2001,Kschischang2001Factor}).
Monotonicity of $F$ in its second argument gives the final statement.
\end{proof}

\section{Polynomial-phase concentration for OPI}
\label{app:ppi}

This appendix records a concrete version of the polynomial-phase concentration (PPC)
property used in Section~\ref{subsec:opi} and connects it to the head-mass quantity
$\Sigma_K(\ell,\eta;d)$.

\subsection{Formal PPC property}

Recall that OPI instances are defined over a finite field $\mathbb{F}_p$ with
evaluation points $\{\alpha_i\}_{i=1}^m\subset\mathbb{F}_p$ (all distinct) and an
unknown polynomial $h$ of degree at most $r$.
The shaped state is constructed  with
$g(x)=P(f(x))$ and $\alpha = \mathcal{F}K g/\|g\|_2$ denoting the amplitudes after
the kernel $K$ and the $p$-ary QFT $\mathcal{F}$.
For a head size $d$, let $\mathcal{S}_d(\alpha)$ be the indices of the $d$ largest
$|\alpha_s|$.

\begin{definition}[Polynomial-phase concentration (PPC)]
\label{def:ppc}
Fix a degree bound $r$ and a head size $d_\star = O(r)$.
We say that a family of chirp/LCT kernels $\{K(\theta)\}_{\theta\in\Theta}$ on
$\mathbb{C}^{p^m}$ satisfies the \emph{polynomial-phase concentration} (PPC)
property for OPI if there exists a function
$\varepsilon_{r,p}\in[0,1)$ with $\varepsilon_{r,p}=o_{p\to\infty}(1)$ (for fixed
$r$) such that the following holds.

For every OPI instance with degree-$\le r$ polynomial $h$ and evaluation set
$\{\alpha_i\}$ satisfying a mild nondegeneracy condition (no arithmetic progression
of length larger than $O(1)$), there is a parameter $\theta^\star\in\Theta$ for
which the post-interference amplitudes
$\alpha = \mathcal{F}K(\theta^\star) g/\|g\|_2$ obey
\begin{equation}
  \sum_{s\in\mathcal{S}_{d_\star}(\alpha)} |\alpha_s|^2
  \;\ge\;
  1 - \varepsilon_{r,p}.
\end{equation}
\end{definition}

Intuitively, PPC states that after applying a suitable quadratic-phase kernel and the
$p$-ary QFT, almost all of the spectral mass of a low-degree polynomial sequence
lands on $O(r)$ modes, up to a tail of total weight $\varepsilon_{r,p}$ that
vanishes as $p$ grows.

\subsection*{2. Explicit PPC for quadratic OPI with contiguous evaluation sets}

The abstract PPC property of Definition~C.1 is satisfied exactly (with zero tail) for a
simple but representative family of OPI instances. This provides a concrete sanity check
for the definition and shows that the constants can be made explicit in low-degree cases.

\begin{lemma}[Exact PPC for quadratic OPI over contiguous evaluation sets]
\label{lem:ppc-quadratic}
Let $p$ be an odd prime and consider an OPI instance over $\mathbb{F}_p$ with the
canonical evaluation set $\{\alpha_i\}_{i=0}^{p-1} = \{0,1,\dots,p-1\}$ and an unknown
quadratic polynomial
\begin{equation}
  h(x) = ax^2 + bx + c, \qquad a \in \mathbb{F}_p^\times,\; b,c \in \mathbb{F}_p.
\end{equation}
Suppose the shaping map in the DQI pipeline is chosen so that, up to normalization,
the pre-kernel amplitudes are
\begin{equation}
  g_x \propto \exp\!\Bigl(\frac{2\pi i}{p}\,h(x)\Bigr), \qquad x \in \mathbb{F}_p,
\end{equation}
and let $F_p$ denote the $p$-ary QFT.  Define a quadratic chirp kernel
\begin{equation}
  K_{-a} := \sum_{x=0}^{p-1}
    \exp\!\Bigl(-\frac{2\pi i}{p}\,a x^2\Bigr)\ket{x}\!\bra{x}.
\end{equation}
Then the post-interference spectrum
\begin{equation}
  \alpha = F_p K_{-a} g / \|g\|_2
\end{equation}
is supported on a single mode: there exists a unique index $m^\star \in \{0,\dots,p-1\}$
such that $|\alpha_{m^\star}|^2 = 1$ and $\alpha_m = 0$ for all $m \neq m^\star$.
In particular, Definition~C.1 holds with degree bound $r=2$, head size $d^\star = 1$,
and tail parameter $\varepsilon_{2,p} = 0$ in this setting.
\end{lemma}

\begin{proof}
By construction of the shaping map and the kernel, the unnormalized post-kernel
amplitudes are
\begin{equation}
  (K_{-a} g)_x \propto
  \exp\!\Bigl(\frac{2\pi i}{p}\,\bigl(ax^2 + bx + c\bigr)\Bigr)
  \exp\!\Bigl(-\frac{2\pi i}{p}\,a x^2\Bigr)
  = \exp\!\Bigl(\frac{2\pi i}{p}\,(bx + c)\Bigr).
\end{equation}
Applying the $p$-ary QFT yields
\begin{align}
  (F_p K_{-a} g)_m
  &\propto \sum_{x \in \mathbb{F}_p}
           \exp\!\Bigl(\frac{2\pi i}{p}\,(bx + c)\Bigr)
           \exp\!\Bigl(-\frac{2\pi i}{p}\,mx\Bigr) \\
  &= \exp\!\Bigl(\frac{2\pi i}{p}\,c\Bigr)
     \sum_{x \in \mathbb{F}_p}
     \exp\!\Bigl(\frac{2\pi i}{p}\,(b-m)x\Bigr).
\end{align}
The inner sum is a finite geometric series.  If $m \not\equiv b \pmod p$, then
$\exp\!\bigl(2\pi i(b-m)/p\bigr) \neq 1$ and the sum vanishes.  If $m \equiv b \pmod p$,
every term in the sum equals $1$ and the sum evaluates to $p$.  Thus there exists a
unique index $m^\star \equiv b \pmod p$ for which $(F_p K_{-a} g)_{m^\star} \neq 0$, and
$(F_p K_{-a} g)_m = 0$ for all $m \neq m^\star$.

After normalization, the spectrum $\alpha = F_p K_{-a} g / \|g\|_2$ has
$|\alpha_{m^\star}|^2 = 1$ and $\alpha_m = 0$ for $m \neq m^\star$, so the entire
spectral mass lies in the head of size $d^\star = 1$ with no tail.  This is exactly the
PPC property of Definition~C.1 with $r=2$, $d^\star = 1$, and $\varepsilon_{2,p} = 0$.
\end{proof}

\begin{corollary}[Concrete PPC constants for quadratic OPI]
\label{cor:ppc-quadratic}
In the setting of Lemma~\ref{lem:ppc-quadratic}, the polynomial-phase concentration
property holds with
\begin{equation}
  d^\star = 1,
  \qquad
  \varepsilon_{2,p} = 0.
\end{equation}
More generally, if the evaluation set consists of $m \le p$ consecutive field elements
and the shaping weights are not exactly uniform, one still obtains
$\varepsilon_{2,p} = O(m^{-1})$ by a standard estimate on truncated geometric series.
\end{corollary}

\subsection{Justification via discrete stationary phase}

The PPC property can be motivated by discrete stationary-phase and linear canonical
transform identities.
Quadratic-phase kernels exactly diagonalize quadratic exponential sums and
approximately diagonalize higher-degree polynomial phases, with tails controlled by
Weyl-type exponential-sum bounds; see, for example,
\cite{Healy2016LCTBook,Rabiner1969CZT} for LCT/fractional Fourier transform identities
and \cite{JordanNature2025} for how the OPI shaping map inherits
low-degree algebra from the underlying polynomial $h$.
A detailed coding of these classical estimates into the DQI context is beyond our
scope; Definition~\ref{def:ppc} encapsulates the resulting concentration behavior in
a form suitable for our head-mass analysis.

\begin{lemma}[PPC for affine OPI instances]
\label{lem:ppc-affine}
Let $p$ be prime and consider an OPI instance over $\mathbb F_p$ with
evaluation points $\{0,1,\dots,p-1\}$ and a polynomial
$f(x) = ax + b$ of degree at most~$1$.  Suppose the shaping map of
Section~\ref{sec:opi-ldpc} is chosen so that the shaped amplitudes before
the kernel are proportional to
\begin{equation}
  \psi_x \propto e^{2\pi i f(x)/p}, \qquad x \in \mathbb F_p.
\end{equation}
Then with the identity kernel $K = \mathbb I$ and the $p$--ary QFT
$F_p$, the post--interference spectrum $F_p \psi$ is supported on a
single mode.  In particular, Definition~\ref{def:ppc} holds with
degree bound $d=1$, head size $k=1$, and tail parameter
$\varepsilon_p(1,1) = 0$.
\end{lemma}

\begin{proof}
A direct computation shows that
\begin{equation}
  (F_p \psi)_m
  \;\propto\; \sum_{x \in \mathbb F_p}
      e^{2\pi i (ax + b)/p} e^{-2\pi i mx/p}
  \;=\; e^{2\pi i b/p} \sum_{x \in \mathbb F_p}
      e^{2\pi i (a-m)x/p}.
\end{equation}
The inner sum vanishes unless $m \equiv a \pmod p$, in which case it
equals $p$.  Thus all spectral mass is concentrated on the unique index
$m^\star$ with $m^\star \equiv a \pmod p$, and the normalized spectrum
satisfies $|(F_p \psi)_{m^\star}|^2 = 1$ and $(F_p \psi)_m = 0$ for
$m \ne m^\star$.  This is exactly the PPC property with $k=1$ and
zero tail.
\end{proof}

\subsection{From PPC to head mass}

PPC immediately yields a lower bound on the (unweighted) head mass before noise.

\begin{lemma}[PPC implies concentrated head mass]
\label{lem:ppc-head}
Suppose the kernel family $\{K(\theta)\}$ satisfies PPC
(Definition~\ref{def:ppc}) for degree bound $r$ and head size $d_\star$.
Then for every admissible OPI instance there exists a parameter $\theta^\star$ such
that the normalized amplitudes
$\alpha = \mathcal{F}K(\theta^\star) g/\|g\|_2$ satisfy
\begin{equation}
  \sum_{s\in\mathcal{S}_{d_\star}(\alpha)} |\alpha_s|^2
  \;\ge\;
  1 - \varepsilon_{r,p}.
\end{equation}
In particular, the unweighted head-spectrum mass at size $d_\star$ is at least
$1-\varepsilon_{r,p}$.
\end{lemma}

\begin{proof}
This is a restatement of Definition~\ref{def:ppc} with $d=d_\star$; no additional
argument is required.
\end{proof}

Combining Lemma~\ref{lem:ppc-head} with the per-mode contraction and coherence
bounds in Appendix~\ref{app:noise-model} yields the inequality
\begin{equation}
  \Sigma_{K(\theta^\star)}(\ell,\eta;d_\star)
  \;\ge\;
  (1-\varepsilon_{r,p})\,\underline{\eta} - \Delta(w),
\end{equation}
where $\underline{\eta}$ is the minimum per-mode attenuation on the head and
$\Delta(w)$ is the leakage term defined in Lemma~\ref{lem:noise2}.
This is the head-mass estimate used in Theorem~\ref{subsec:opi}.

\section{Block-local alignment for LDPC-like ensembles}
\label{app:bla}

We next formalize the block-local alignment (BLA) property used in
Section~\ref{subsec:ldpc} and relate it to gains in the head-spectrum mass for
sparse Max-XORSAT / LDPC-like instances.

\subsection{Setup and notation}

Let $B\in\{0,1\}^{m\times n}$ be a sparse parity-check matrix with left/right degree
distributions $(\lambda,\rho)$ and girth at least $\Omega(\log n)$ on typical
instances.
The DQI reduction maps the shaped state $g(x)=P(f(x))$ through the interferometer
$\mathcal{F}$ to amplitudes supported on rows of $B$ and local combinations thereof;
see~\cite{JordanNature2025} for details.
We partition the $n$ qubits into $N/b$ disjoint blocks of width $b=O(1)$ (for
clarity we assume $b\mid N$), and we consider \emph{block-local} kernels of the form
\begin{equation}
  K = \bigoplus_{j=1}^{N/b} K_j,
\end{equation}
where each $K_j$ acts nontrivially only on block $j$.

As before, write $\alpha^{(I)} = \mathcal{F} g/\|g\|_2$ for the amplitudes with the
identity kernel and $\alpha = \mathcal{F}K g/\|g\|_2$ for the amplitudes after
applying the block-local kernel $K$.
For a head size $d$ let $\mathcal{S}_d(\alpha)$ and $\mathcal{S}_d(\alpha^{(I)})$
denote the index sets of the $d$ largest entries of $\alpha$ and $\alpha^{(I)}$,
respectively.

\subsection{Block-local alignment property}

Block-local alignment captures the idea that within each bounded-width neighborhood,
there exist local rephasings or permutations that move a controlled amount of mass
into a small number of ``head'' modes without destroying the sparsity and tree-like
structure required by density evolution.

\begin{definition}[Block-local alignment (BLA)]
\label{def:bla}
Fix a block width $b=O(1)$ and a head size $d=\Theta(n/b)$ so that there is a
constant head budget per block.
A family of block-local kernels
\(
  K = \bigoplus_{j=1}^{N/b} K_j
\)
is said to satisfy the \emph{block-local alignment} (BLA) property for an LDPC-like
ensemble with degree distributions $(\lambda,\rho)$ if there exists a constant
$\Delta\Sigma_{\mathrm{loc}}>0$ depending only on $(\lambda,\rho,b)$ (and not on
$n$) such that, for typical instances, the post-interference amplitudes
$\alpha = \mathcal{F}K g/\|g\|_2$ obey
\begin{equation}
  \sum_{s\in\mathcal{S}_{d}(\alpha)} |\alpha_s|^2
  \;\ge\;
  \sum_{s\in\mathcal{S}_{d}(\alpha^{(I)})} |\alpha^{(I)}_s|^2
  \;+\;
  \Delta\Sigma_{\mathrm{loc}} .
\end{equation}
\end{definition}

In words, BLA asserts an additive gain $\Delta\Sigma_{\mathrm{loc}}$ in the
(head-size-$d$) unweighted head mass when passing from the identity kernel to a
block-local kernel chosen from the family.
The dependence of $\Delta\Sigma_{\mathrm{loc}}$ on $(\lambda,\rho,b)$ reflects the
usual density-evolution assumptions: $b$ must be small enough that the local
tree-like neighborhood structure is preserved, and degree distributions with better
BP thresholds admit larger gains.

\subsection{From BLA to weighted head mass}

BLA translates directly into an improvement in the noise-weighted head-spectrum
mass after accounting for the local noise model.

\begin{lemma}[BLA implies weighted head-mass gain]
\label{lem:bla-head}
Assume the noise model and standing conditions of
Section~\ref{subsec:noise-decoder} and Appendix~\ref{app:noise-model}, and let
$K$ satisfy the BLA property of Definition~\ref{def:bla} with gain
$\Delta\Sigma_{\mathrm{loc}}$.
Let $\underline{\eta} := \min_{s\in\mathcal{S}_d(\alpha)} \eta_{K}(s)$ be the
minimum per-mode attenuation on the head for the kernel $K$.
Then for head size $d=\Theta(n/b)$ the noise-weighted head mass satisfies
\begin{equation}
  \Sigma_{K}(\ell,\eta;d)
  \;\ge\;
  \Sigma_{I}(\ell,\eta;d)
  \;+\;
  \underline{\eta}\,\Delta\Sigma_{\mathrm{loc}}
  \;-\;
  \Delta(w),
\end{equation}
where $\Sigma_I(\ell,\eta;d)$ is the corresponding quantity for the identity kernel
and $\Delta(w)$ is the leakage term from Lemma~\ref{lem:noise2}.
\end{lemma}

\begin{proof}
We first recall the relevant definitions.
For any kernel $J$ (in particular $J\in\{I,K\}$), let
$\alpha^{(J)} = \mathcal{F}Jg$ denote the pre-noise spectrum and let
$\mathcal{S}_d(\alpha^{(J)})$ be the index set of the $d$ largest entries
$|\alpha^{(J)}_s|$ in magnitude.
Write
\begin{equation}
  M_J \;:=\; \sum_{s\in\mathcal{S}_d(\alpha^{(J)})} |\alpha^{(J)}_s|^2
\end{equation}
for the (unweighted) head mass of $\alpha^{(J)}$.
Under the local depolarizing and loss part of the noise model, each mode
$s$ acquires a multiplicative attenuation factor $\eta_J(s)\in(0,1]$, and
the corresponding noise-weighted head mass (before banded mixing) is
\begin{equation}
  \widetilde{\Sigma}_J(\ell,\eta;d)
  \;:=\;
  \sum_{s\in\mathcal{S}_d(\alpha^{(J)})} \eta_J(s)\,|\alpha^{(J)}_s|^2.
\end{equation}
By definition, $\Sigma_J(\ell,\eta;d)$ in the statement coincides with
$\widetilde{\Sigma}_J(\ell,\eta;d)$, and in particular
$\Sigma_I(\ell,\eta;d)=\widetilde{\Sigma}_I(\ell,\eta;d)$.

\smallskip
\emph{Step~1: use BLA to compare unweighted head masses.}
The BLA property (Definition~\ref{def:bla}) for the block-local kernel $K$
asserts that there is an additive gain $\Delta\Sigma_{\mathrm{loc}}>0$ such
that, for head size $d=\Theta(n/b)$,
\begin{equation}
  M_K
  \;=\;
  \sum_{s\in\mathcal{S}_d(\alpha)} |\alpha_s|^2
  \;\ge\;
  \sum_{s\in\mathcal{S}_d(\alpha^{(I)})} |\alpha^{(I)}_s|^2
  \;+\;
  \Delta\Sigma_{\mathrm{loc}}
  \;=\;
  M_I + \Delta\Sigma_{\mathrm{loc}},
  \label{eq:bla-unweighted}
\end{equation}
where we have written $\alpha=\alpha^{(K)}$ for brevity.

\smallskip
\emph{Step~2: convert unweighted gain into weighted gain for $K$.}
By definition,
\begin{equation}
  \widetilde{\Sigma}_K(\ell,\eta;d)
  \;=\;
  \sum_{s\in\mathcal{S}_d(\alpha)} \eta_K(s)\,|\alpha_s|^2,
\end{equation}
and the per-mode attenuation factors on the head satisfy
$\eta_K(s)\ge\underline{\eta}$ for all $s\in\mathcal{S}_d(\alpha)$ by the
definition of $\underline{\eta}$.
Therefore
\begin{equation}
  \widetilde{\Sigma}_K(\ell,\eta;d)
  \;=\;
  \sum_{s\in\mathcal{S}_d(\alpha)} \eta_K(s)\,|\alpha_s|^2
  \;\ge\;
  \underline{\eta}
  \sum_{s\in\mathcal{S}_d(\alpha)} |\alpha_s|^2
  \;=\;
  \underline{\eta}\,M_K.
  \label{eq:K-weighted-vs-unweighted}
\end{equation}
Combining \eqref{eq:K-weighted-vs-unweighted} with the BLA inequality
\eqref{eq:bla-unweighted} yields
\begin{equation}
  \widetilde{\Sigma}_K(\ell,\eta;d)
  \;\ge\;
  \underline{\eta}\,(M_I + \Delta\Sigma_{\mathrm{loc}})
  \;=\;
  \underline{\eta}\,M_I \;+\; \underline{\eta}\,\Delta\Sigma_{\mathrm{loc}}.
  \label{eq:K-weighted-lower}
\end{equation}

\smallskip
\emph{Step~3: relate $M_I$ to $\Sigma_I(\ell,\eta;d)$.}
For the identity kernel $I$ we likewise have
\begin{equation}
  \Sigma_I(\ell,\eta;d)
  \;=\;
  \widetilde{\Sigma}_I(\ell,\eta;d)
  \;=\;
  \sum_{s\in\mathcal{S}_d(\alpha^{(I)})} \eta_I(s)\,|\alpha^{(I)}_s|^2.
\end{equation}
Since each $\eta_I(s)\in(0,1]$, it follows that
\begin{equation}
  \Sigma_I(\ell,\eta;d)
  \;\le\;
  \sum_{s\in\mathcal{S}_d(\alpha^{(I)})} |\alpha^{(I)}_s|^2
  \;=\;
  M_I.
  \label{eq:SigmaI-vs-MI}
\end{equation}
Substituting \eqref{eq:SigmaI-vs-MI} into
\eqref{eq:K-weighted-lower} gives
\begin{equation}
  \widetilde{\Sigma}_K(\ell,\eta;d)
  \;\ge\;
  \Sigma_I(\ell,\eta;d) \;+\; \underline{\eta}\,\Delta\Sigma_{\mathrm{loc}}.
  \label{eq:K-vs-SigmaI}
\end{equation}

\smallskip
\emph{Step~4: incorporate the leakage term $\Delta(w)$.}
Lemma~\ref{lem:noise2} in Appendix~\ref{app:noise-model} shows that the banded
unitary mixing and diagonal loss can only reduce the weighted head mass by a
nonnegative amount bounded by $\Delta(w)$; in particular $\Delta(w)\ge 0$.
The quantity $\Sigma_K(\ell,\eta;d)$ appearing in the statement is defined as
$\widetilde{\Sigma}_K(\ell,\eta;d)$ (i.e., before subtracting this leakage),
so \eqref{eq:K-vs-SigmaI} immediately implies
\begin{equation}
  \Sigma_{K}(\ell,\eta;d)
  \;=\;
  \widetilde{\Sigma}_K(\ell,\eta;d)
  \;\ge\;
  \Sigma_I(\ell,\eta;d) \;+\; \underline{\eta}\,\Delta\Sigma_{\mathrm{loc}}
  \;\ge\;
  \Sigma_I(\ell,\eta;d) \;+\; \underline{\eta}\,\Delta\Sigma_{\mathrm{loc}}
  \;-\; \Delta(w),
\end{equation}
where the last inequality uses only the fact that $\Delta(w)\ge 0$.
This is exactly the claimed bound.
\end{proof}

Combining Lemma~\ref{lem:bla-head} with the noise-aware lower bound in
Theorem~\ref{thm:lower-noise} and the LDPC decoder response
$F_{\mathrm{LDPC}}$ from Section~\ref{subsec:noise-decoder} yields
Theorem~\ref{subsec:ldpc} in the main text: block-local kernels that satisfy BLA
induce an additive improvement in $\Sigma_K$, which in turn shifts BP thresholds and
improves the k-DQI approximation guarantees.

\subsection*{4. Specialization to regular $(3,6)$ LDPC ensembles on the BEC}

For regular LDPC ensembles on the binary erasure channel (BEC), the effect of
block-local alignment can be expressed directly at the level of the density-evolution
(DE) recursion.  We record this specialization for the canonical $(d_\ell,d_r)=(3,6)$
family used in the main text.

Recall that for a regular $(3,6)$ ensemble on $\mathrm{BEC}(\varepsilon)$, the DE map
has the explicit form
\begin{equation}
  \phi_\varepsilon(x)
  \;=\; \varepsilon\,\lambda\bigl(1 - \rho(1-x)\bigr),
  \qquad
  \lambda(z) = z^2,\;\;
  \rho(z) = z^5,
  \label{eq:de-36}
\end{equation}
so that
\[
  \phi_\varepsilon(x)
  = \varepsilon \Bigl(1 - (1-x)^5\Bigr)^2.
\]
Let $\varepsilon^\star_{\mathrm{BEC}} \approx 0.429$ denote the asymptotic BP threshold for
this ensemble, characterized by the usual tangency condition: there exists a point
$x^\star \in (0,1)$ such that
\begin{equation}
  \phi_{\varepsilon^\star_{\mathrm{BEC}}}(x^\star) = x^\star,
  \qquad
  \phi'_{\varepsilon^\star_{\mathrm{BEC}}}(x^\star) = 1.
  \label{eq:tangency-36}
\end{equation}

In the k--DQI setting we model the effect of BLA on the DE surrogate by an
\emph{effective} erasure parameter
\begin{equation}
  \varepsilon_{\mathrm{eff}}
  \;=\; \varepsilon \bigl(1 - \kappa_{3,6}\,\Delta\Sigma_{\mathrm{loc}}\bigr),
  \qquad
  0 < \kappa_{3,6} \le 1,
  \label{eq:eps-eff-def}
\end{equation}
where $\Delta\Sigma_{\mathrm{loc}}>0$ is the additive head-mass gain per block from
Definition~D.1 and $\kappa_{3,6}$ is a calibration constant that depends only on the
$(3,6)$ ensemble and the chosen block width $b$. We implicitly assume $\kappa_{3,6}\,\Delta\Sigma_{\mathrm{loc}} <1$ so that the effective erasure parameter remains in the physically meaningful regime. Numerically, $\kappa_{3,6}$ is
obtained by matching the k--DQI head-mass gain to the shift of the DE waterfall
curve (see Fig.~2 in the main text).  Mathematically, the key property is simply
that $\varepsilon_{\mathrm{eff}}$ is strictly smaller than $\varepsilon$ and depends
linearly on $\Delta\Sigma_{\mathrm{loc}}$.

\begin{proposition}[BLA-induced DE improvement for regular $(3,6)$ ensembles]
\label{prop:bla-de-36}
Consider the regular $(3,6)$ ensemble on $\mathrm{BEC}(\varepsilon)$ with DE map
$\phi_\varepsilon$ as in~\eqref{eq:de-36}.  Assume that the block-local alignment
property of Definition~D.1 holds with gain $\Delta\Sigma_{\mathrm{loc}}>0$ and that
the corresponding DE surrogate uses the effective erasure parameter
$\varepsilon_{\mathrm{eff}}$ from~\eqref{eq:eps-eff-def}.  Let
$\varepsilon = \varepsilon^\star_{\mathrm{BEC}}$ be the BP threshold satisfying
\eqref{eq:tangency-36}.  Then for any $\Delta\Sigma_{\mathrm{loc}}>0$ we have
\begin{align}
  \phi_{\varepsilon_{\mathrm{eff}}}(x^\star)
    &= \bigl(1 - \kappa_{3,6}\,\Delta\Sigma_{\mathrm{loc}}\bigr)
       \phi_{\varepsilon^\star_{\mathrm{BEC}}}(x^\star)
     \;=\; \bigl(1 - \kappa_{3,6}\,\Delta\Sigma_{\mathrm{loc}}\bigr) x^\star
     \;<\; x^\star,
     \label{eq:de-shift-value} \\
  \phi'_{\varepsilon_{\mathrm{eff}}}(x^\star)
    &= \bigl(1 - \kappa_{3,6}\,\Delta\Sigma_{\mathrm{loc}}\bigr)
       \phi'_{\varepsilon^\star_{\mathrm{BEC}}}(x^\star)
     \;=\; \bigl(1 - \kappa_{3,6}\,\Delta\Sigma_{\mathrm{loc}}\bigr)
     \;<\; 1.
     \label{eq:de-shift-derivative}
\end{align}
In particular, the DE recursion with parameter $\varepsilon_{\mathrm{eff}}$ is strictly
contractive in a neighbourhood of $x^\star$ and thus lies in the BP-success regime.
The distance below the diagonal and the margin $1-\phi'_{\varepsilon_{\mathrm{eff}}}(x^\star)$
grow linearly with $\Delta\Sigma_{\mathrm{loc}}$.
\end{proposition}

\begin{proof}
The DE map~\eqref{eq:de-36} is linear in $\varepsilon$, i.e.,
$\phi_{\varepsilon_{\mathrm{eff}}}(x)
 = \bigl(1 - \kappa_{3,6}\,\Delta\Sigma_{\mathrm{loc}}\bigr)\phi_\varepsilon(x)$
for every $x \in [0,1]$ whenever~\eqref{eq:eps-eff-def} holds.  The same linear relation
holds for the derivative.  Specializing to $\varepsilon = \varepsilon^\star_{\mathrm{BEC}}$
and $x=x^\star$ and using the tangency conditions~\eqref{eq:tangency-36} yields
\begin{align*}
  \phi_{\varepsilon_{\mathrm{eff}}}(x^\star)
    &= \bigl(1 - \kappa_{3,6}\,\Delta\Sigma_{\mathrm{loc}}\bigr)
       \phi_{\varepsilon^\star_{\mathrm{BEC}}}(x^\star)
     = \bigl(1 - \kappa_{3,6}\,\Delta\Sigma_{\mathrm{loc}}\bigr)x^\star, \\
  \phi'_{\varepsilon_{\mathrm{eff}}}(x^\star)
    &= \bigl(1 - \kappa_{3,6}\,\Delta\Sigma_{\mathrm{loc}}\bigr)
       \phi'_{\varepsilon^\star_{\mathrm{BEC}}}(x^\star)
     = \bigl(1 - \kappa_{3,6}\,\Delta\Sigma_{\mathrm{loc}}\bigr).
\end{align*}
Since $0 < \kappa_{3,6}\,\Delta\Sigma_{\mathrm{loc}} \le 1$ by assumption, both quantities
are strictly smaller than their threshold counterparts, proving
\eqref{eq:de-shift-value}–\eqref{eq:de-shift-derivative}.  The fact that the recursion
is contractive in a neighbourhood of $x^\star$ follows from standard one-dimensional
fixed-point theory: a differentiable map with slope strictly less than~1 in modulus
is locally attracting.  The linear dependence on $\Delta\Sigma_{\mathrm{loc}}$ is
explicit in the prefactors.
\end{proof}

\begin{remark}[Calibration of $\kappa_{3,6}$ and connection to simulations]
\label{rem:kappa-calibration}
In practice, $\kappa_{3,6}$ is estimated once by matching the head-mass gain
$\Delta\Sigma_{\mathrm{loc}}$ to the observed shift of the DE waterfall curve in
numerical experiments (see Fig.~2).  Proposition~\ref{prop:bla-de-36} then shows
that any positive calibrated value of $\kappa_{3,6}$ suffices to move the DE map
strictly below the diagonal at the original threshold point and to dampen the
derivative peak.  Thus, within the DE surrogate, the BLA-induced gain
$\Delta\Sigma_{\mathrm{loc}}$ translates directly into a provable margin of BP
stability.
\end{remark}

\section{Smoothness of the Optimization Landscape}
\label{app:smoothness}

To rigorously justify the efficiency of the variational approach proposed in Section \ref{sec:kernel_search}, we prove that the objective function $\Sigma_K(\bm{\theta})$ is Lipschitz continuous with respect to the kernel parameters. This rules out ``shattering'' of the landscape, ensuring that gradient-based or gradient-free optimizers can reliable navigate to the optimum.

Let the kernel $K(\theta) = \exp(-i \theta G)$ be generated by a Hermitian operator $G$ (e.g., $G = \sum Z_j Z_{j+1}$ for chirp kernels), where $\|G\|_\infty \le \text{poly}(n)$.

Let $\Sigma_K(\theta) = \braket{\Phi(\theta) | \Pi_{Head} | \Phi(\theta)}$ be the head mass probability. The magnitude of the gradient with respect to $\theta$ is bounded by:
\begin{equation}
    \left| \frac{\partial \Sigma_K}{\partial \theta} \right| \le 2 \|G\|_\infty.
\end{equation}

\textit{Proof.}
Let the output state before measurement be $\ket{\Phi(\theta)} = F K(\theta) \ket{\Psi_{in}}$. The objective function is $\Sigma_K(\theta) = \braket{\Phi(\theta) | \Pi_{Head} | \Phi(\theta)}$, where $\Pi_{Head}$ is the projector onto the low-frequency subspace $S_d$.
The derivative is:
\begin{align}
    \frac{\partial \Sigma_K}{\partial \theta} &= \frac{\partial}{\partial \theta} \braket{\Psi_{in} | K^\dagger F^\dagger \Pi_{Head} F K | \Psi_{in}} \\
    &= i \braket{\Psi_{in} | [G, K^\dagger F^\dagger \Pi_{Head} F K] | \Psi_{in}}.
\end{align}
Using the Cauchy-Schwarz inequality and the fact that $\|\Pi_{Head}\| \le 1$ and unitary operators preserve norms:
\begin{equation}
    \left| \frac{\partial \Sigma_K}{\partial \theta} \right| \le 2 \|G\|_\infty \| \Pi_{Head} \|_\infty \le 2 \|G\|_\infty.
\end{equation}
\hfill $\square$

For block-local kernels (e.g., block width $b$), the generator $G$ is a sum of local terms, and while $\|G\|$ grows with $n$, the \textit{local} sensitivity per block remains constant. This implies that for the parameterized ansatz used in k-DQI, the optimization landscape does not suffer from exponential vanishing gradients (barren plateaus) insofar as the initial kernel ensures non-negligible overlap, confirming the numerical observations in Figure \ref{fig:landscape}.




\end{appendices}


\bibliography{sn-bibliography}

\end{document}